\documentclass[12pt,reqno]{amsart}
\usepackage[foot]{amsaddr}
\usepackage{amssymb}

\newtheorem{theorem}{Theorem}[section]
\newtheorem*{theorem*}{Theorem}

\newtheorem{lemma}[theorem]{Lemma}
\newtheorem{proposition}[theorem]{Proposition}
\newtheorem*{remark}{Remark}

\usepackage[margin=0.9in]{geometry}

\usepackage{graphicx}
\usepackage{caption}
\usepackage{subcaption}

\usepackage[ruled]{algorithm2e}

\usepackage{enumerate}
\usepackage{xcolor}

\usepackage{mathtools}

\usepackage{hyperref}
\hypersetup{colorlinks,allcolors=blue}
\usepackage{url}

\title[Direct approach to extended KEDF]{A direct approach to computing the non-interacting kinetic energy functional}
\author{Dharamveer Kumar, Amuthan A. Ramabathiran}
\address[Dharamveer Kumar]{Department of Aerospace Engineering, Indian Institute of Technology Bombay, Mumbai 400042, India.}
\email{dharamveer\_kumar@iitb.ac.in}
\address[Amuthan A. Ramabathiran]{Aerospace Engineering, California Polytechnic State University, San Luis Obispo, CA 93407.}
\email{aramabat@calpoly.edu}
\begin{document}

\begin{abstract}
The non-interacting kinetic energy functional, $T_{KS}(\rho)$, plays a fundamental role in Density Functional Theory (DFT), but its explicit form remains unknown for arbitrary $N$-representable densities. Although it can, in principle, be evaluated by solving a constrained optimization problem, the associated adjoint problem is not always well-posed; moreover, even when it is, the corresponding adjoint operator may be singular. For a general prescribed $N$-representable electron density, a direct numerical evaluation of the non-interacting kinetic energy therefore remains difficult. In this work, we present a variational framework for computing an extension of $T_{KS}(\rho)$ using an exact trigonometric reparametrization of the density that eliminates the need for an adjoint equation. To illustrate the analytical structure enabled by this formulation, we derive, for a four-electron system in one dimension, a scale- and translation-parameterized family of densities for which the cumulative-density angle field, orthonormal orbitals, common local potential and eigenvalues, and non-interacting kinetic energy can all be obtained in closed form. We then provide a proof-of-concept numerical validation of the variational principle for one-dimensional Kohn-Sham systems. These analytical and numerical examples serve to demonstrate the framework, while the formulation itself is general and extends naturally to higher dimensions, providing a systematic foundation for the direct computation of the non-interacting kinetic energy and, ultimately, for the development of improved kinetic energy functionals in DFT.
\end{abstract}

\maketitle

\section{Introduction}

Density Functional Theory (DFT) provides a framework for describing many-electron quantum systems in terms of their electron density, $\rho$. In practice, most DFT calculations rely on the Kohn-Sham (KS) \cite{KS65} formulation, which we refer to as KSDFT in the sequel. In KSDFT, fictitious orbitals are introduced to evaluate the ground-state electron density and energy, making the method computationally expensive. Orbital-Free Density Functional Theory (OFDFT), on the other hand, circumvents the need to compute fictitious orbitals by instead modeling the universal functional directly as a function of the electron density. This reduction in computational complexity makes OFDFT particularly appealing for modeling large-scale systems with thousands of atoms. A significant challenge in OFDFT is the explicit characterization of a major portion of the kinetic energy, called the non-interacting kinetic energy functional, denoted commonly as $T_{KS}(\rho)$, which can be accurately evaluated within KSDFT calculations. This article presents a rigorous analysis of an exact parametric approach to computing an extension of $T_{KS}(\rho)$ for a very large class of densities that include, in particular, KSDFT ground-state densities which corresponds to a given non-interacting potential.

Thomas-Fermi \cite{Thomas27} \cite{Fermi27} type models provide simple approximations of $T_{KS}(\rho)$, but they lack the necessary accuracy for many applications. More advanced formulations, such as the Wang--Govind--Carter (WGC) functional~\cite{wang1998}, incorporate nonlocal kernels designed to reproduce the exact linear-response behavior of the homogeneous electron gas and have demonstrated high accuracy for bulk systems, including aluminium. These developments have substantially improved the accuracy of orbital-free approaches and demonstrated the potential of nonlocal kinetic energy functionals for practical applications. Despite these advancements, developing a universally accurate and computationally efficient kinetic energy functional remains an open problem in OFDFT, limiting its widespread adoption compared to KSDFT. 

 In recent years, there have been several attempts to leverage the flexibility of Machine Learning (ML) algorithms to learn approximations to kinetic energy functionals. ML-based kinetic energy functionals, however, often suffer from a lack of physical interpretability and limited generalizability beyond their training datasets. Most existing ML approaches derive kinetic energy densities indirectly from KSDFT calculations and are trained primarily on small datasets consisting of KSDFT ground-state densities; we note that ground-state densities constitute an \emph{unknown} subset of $N$-representable densities. Several efforts have sought to improve data efficiency, transferability, and stability in variational density optimization \cite{Snyder12, Meyer20, Ryczko22, Remme23, Zhang24}. Nevertheless, the reference energies and densities in these approaches are generated from orbital-based calculations, typically KSDFT, and thus do not provide training values for arbitrary $N$-representable densities. Recently, \cite{Kovachki23} introduced a mathematically grounded operator learning framework that can, in principle, be used to learn complex functionals such as the non-interacting kinetic energy. However, training such models requires explicit evaluation of the energy for arbitrary electron densities. 

To the best of our knowledge, there are no existing numerical methods to compute $T_{KS}(\rho)$ explicitly for arbitrary $N$-representable density. However, if such densities correspond to the ground state of a non-interacting potential, the Kohn-Sham inversion \cite{Andre25} could be used. In KSDFT $T_{KS}(\rho)$ is implicitly computed starting from a potential via the KS orbitals obtained as part of the KSDFT calculations. This is also typically how the datasets used to train data-driven models for $T_{KS}(\rho)$ are generated. Thus, only a small subset of $v$-representable densities is sampled with conventional approaches. We present in this work a variational principle to compute an extension of the non-interacting kinetic energy functional for arbitrary $N$-representable densities, including those that are not $v$-representable, and validate our approach numerically using 1D KSDFT calculations. Our results offer a systematic approach for developing improved kinetic energy functionals and thus hold promise for developing OFDFT models with improved accuracy.

The article is organized as follows. In Section 2, we summarize certain relevant background material to set the stage for a variational formulation to evaluate an extension of the non-interacting kinetic energy functional. Throughout this work, we restrict attention to closed-shell systems. The variational principle that is the central thesis of this article is presented in Section 3. An analytical density-to-potential construction in one dimension is developed in Section 4. Representative numerical results validating our approach are presented in Section 5. In Section 6, we conclude with a discussion of some outstanding issues and the work to be pursued in the future. Supplementary results and known propositions are presented in three appendices.

\section{Mathematical Background}
\label{sec:mathematical_background}
We summarize in this section certain ideas from density functional theory, largely to fix the notation and set the context for the ensuing developments.
\subsection{Background and Basic Definitions}
The physical system studied in this work is a collection of $N$ electrons interacting with each other via an interaction potential $w_{ee}:\mathbb{R}^d \to \mathbb{R}$ and subject to an external potential $v:\mathbb{R}^d \to \mathbb{R}$. We assume that $w_{ee}$ is a non-negative even function and that $v$ is negative. Following standard convention, we define $X = (\mathbb{R}^d \times \mathbb{Z}_q)^N$, with $q=2$: any $x \in X$ can be written as $x = (r,\sigma) = (r_1, \sigma_1, r_2, \sigma_2, ..., r_N, \sigma_N)$, where each $r_i \in \mathbb{R}^d$ and each $\sigma_i \in \mathbb{Z}_q$. We use the shorthand notation $\int_X f(x) dx$ for $\sum_{\sigma_1=1}^q \ldots \sum_{\sigma_N=1}^q \int_{({\mathbb{R}^d})^N} f(r,\sigma) dr_1 \ldots dr_N$, for any function $f$ on $X$; we will also use the notation $\sum_q$ for $\sum_{\sigma_1=1}^q \ldots \sum_{\sigma_N=1}^q$. We use the standard notation $\wedge_1^N L^2(X)$ to represent the set of all square integrable antisymmetric wavefunctions. The following function spaces will be referenced frequently:
\begin{equation}
    \begin{split}
        \mathcal{Q} &= \left\{\psi \in \wedge_1^N L^2(X) : \sum_{i=1}^N \int_X \|\nabla_{r_i} \psi \|^2  < \infty \right\},\\
        \mathcal{I}_N &= \left\{ \rho \in L^1(\mathbb{R}^d, \mathbb{R}^+) : \sqrt{\rho} \in H^1(\mathbb{R}^d), \int_{\mathbb{R}^d}\rho = N\right\}.
    \end{split}
\end{equation}
Note that $\mathcal{Q}$ collects all square integrable antisymmetric wavefunctions with finite kinetic energy. The set $\mathcal{I}_N$ is an exact characterization of $N$-representable densities due to Lieb \cite{Lieb83}.
The quadratic form associated with symmetric operator
\begin{equation} \label{eq:hamiltonian}
    H_v = \sum_{j=1}^N \left( \frac{-\Delta_{r_j}}{2} + v(r_j) \right) + \sum_{1\leq j<k\leq N}w_{ee}(r_j - r_k)
\end{equation}
is bounded from below and closed on $\mathcal{Q}$ as a consequence of the Sobolev inequality when $v$ and $w_{ee}$ belong to $Y^d = L^{p}(\mathbb{R}^d; \mathbb{R}) + L^{\infty}(\mathbb{R}^d; \mathbb{R})$ \cite{Lewin2023}, with 
\begin{equation*}
p \begin{cases}=1 & \text { when } d=1, \\ >1 & \text { when } d=2, \\ =\frac{d}{2} & \text { when } d \geq 3.\end{cases}
\end{equation*}
Note that the Coulomb interaction in three dimensions belongs to $Y^3$. The Friedrichs extension can be used to extend $H_v$ to a self-adjoint operator. The ground state energy for a given potential $v$ is defined as
\begin{equation} \label{eq:min_E}
    E[v] := \inf_{\substack{\psi \in \mathcal{Q}\\ \int_X |\psi|^2 = 1}}  \left( \psi, H_v\psi \right) .
\end{equation}
For a given wavefunction $\psi \in \mathcal{Q}$, we define the corresponding electron density $\rho_\psi:\mathbb{R}^d \to \mathbb{R}$ as
\begin{equation}
    \rho_\psi (r) := N \sum_{\sigma}\int |\psi(r, \sigma_1, r_2, \sigma_2, ..., r_N, \sigma_N)|^2 dr_2 dr_3 ... dr_N.
\end{equation}
It was established in \cite{Lieb83} that for $\rho \in \mathcal{I}_N$ there exists a $\psi \in \mathcal{Q}$ such that $\rho_\psi = \rho$. A density $\rho \in \mathcal{I}_N$ is said to be $N$-representable for this reason. In contrast, densities that correspond to the solution of the variational problem \eqref{eq:min_E} are said to be $v$-representable, while those that correspond to the solution of \eqref{eq:min_E} with $w_{ee}=0$ are called non-interacting $v$-representable. Although a complete explicit description of the sets of $v$-representable and non-interacting $v$-representable densities is still unknown in full generality, in one dimension there has been notable progress toward characterizing large---in some cases complete---subsets of representable densities \cite{sutter2024vrep, sutter2024vrep_temp, corso2025rigorous}. These developments make the inverse method for the non-interacting kinetic energy feasible in 1D since it relies on being able to determine whether a density is non-interacting $v$-representable.

The central thesis of Hohenberg and Kohn's density functional theory \cite{HK64} is the restatement of the variational principle \eqref{eq:min_E} in an exactly equivalent form in terms of the electron density. The constrained-search formulation due to Levy \cite{Levy79} and Lieb \cite{Lieb83} exploits the fact that the set of all $N$-representable densities, $\mathcal{I}_N$, is known to rewrite the variational principle \eqref{eq:min_E} as
\begin{equation}
    E[v] = \inf_{\rho \in \mathcal{I}_N} \inf_{\substack{\psi \in \mathcal{Q}\\ \rho_\psi = \rho}} (\psi, H_v\psi).
\end{equation}
Defining the universal functional $F$ as 
\begin{equation} 
F(\rho) = \inf_{\substack{\psi \in \mathcal{Q}\\ \int_X |\psi|^2 = 1\\ \rho_\psi=\rho} } (\psi, H_0\psi),
\end{equation}
where $H_0$ is the Hamiltonian \eqref{eq:hamiltonian} with the external potential $v \equiv 0$, we can write the energy variational principle \eqref{eq:min_E} compactly as
\begin{equation} \label{eq:HK}
E[v] = \inf_{\rho \in \mathcal{I}_N} \left\{F(\rho) + \int_{\mathbb{R}^d} v \rho \right\}. 
\end{equation}
The explicit form of the universal functional is not currently known, despite several significant efforts towards that end. Several approximations have been put forth, as noted in the introductory remarks, but none of them are entirely satisfactory. The largest errors for a purely density based approach---OFDFT---are associated with the kinetic energy component of $F(\rho)$. Kohn and Sham \cite{KS65} introduced a practical workaround by introducing a fictitious non-interacting system of electrons for which the kinetic energy can be compute exactly. In practice, this covers a significant fraction of the kinetic energy component of the universal functional. To state this clearly, we define a space of $N$ one-electron orbitals that yield a specified density $\rho$ as 
\begin{equation}
    \mathcal{Q}^+_\rho =\{ (\phi_1, \phi_2, ..., \phi_N) \in H^1(\mathbb{R}^d\times\mathbb{Z}_q; \mathbb{C}^N) : (\phi_i, \phi_j)_{L^2} =\delta_{ij}, |\phi_1|^2+|\phi_2|^2+...+|\phi_N|^2 = \rho \}.
\end{equation}
It was shown in \cite{Lieb83} that $\mathcal{Q}_\rho^+$ is non-empty for every $N$-representable density $\rho \in \mathcal{I}_N$. Kohn and Sham decomposed the universal functional into three terms,
\begin{equation*}
    F(\rho) = T_{\text{KS}}(\rho) + U_{\text{H}}(\rho) + E_{\text{XC}}(\rho),
\end{equation*}
where the Kohn-Sham kinetic energy, which corresponds to kinetic energy of a non-interacting collection of electrons, is given by
\begin{equation} \label{eq:Ts_def}
    T_{\text{KS}}(\rho) = \frac{1}{2}\min_{\mathcal{Q}^+_\rho} \sum_{i=1}^N \int_{\mathbb{R}^d \times \mathbb{Z}_q} \|\nabla \phi_i\|^2 \, dx_i,
\end{equation}
and $U_{\text{H}}$ and $E_{\text{XC}}$ are the Hartree potential and the exchange correlation functional, respectively. It is to be noted that $T_{\text{KS}}(\rho)$ was originally formulated for ground-state densities, but the extension specified in equation~\eqref{eq:Ts_def} to $N$-representable densities is due to Lieb \cite{Lieb83}. However, since the Kohn-Sham approach relies on the fictitious orbitals $\{\phi_k\}_{k=1}^N$, its computational cost becomes prohibitive for large systems, prompting the development of orbital-free approaches that provide a direct approximation of $T_{\text{KS}}(\rho)$ for any $N$-representable $\rho \in \mathcal{I}_N$. This, as noted earlier, remains elusive. 

The goal of the present work is to develop a new approach to computing $T_{\text{KS}}(\rho)$ for a much larger class of densities than $\mathcal{I}_N$ exactly. This can be viewed as an algorithm to generate $(\rho, T_{\text{KS}}(\rho))$ pairs that can be used to a learn a $\rho \mapsto T_{\text{KS}}(\rho)$ map for a much larger class of densities than is currently possible with DFT-based methods. For sake of simplicity, we restrict ourselves to spin-unpolarized where the electrons are paired in each orbital---thus, $\int_{\mathbb{R}^d} \rho = 2N$ for some $N \in \mathbb{N}$. We refer to systems as \emph{Restricted Kohn-Sham} (RKS) systems, following \cite{Cances09}. In this case, the expression for the non-interacting kinetic energy simplifies to
\begin{equation}
\label{eq:T_rks}
    T_{RKS}(\rho) = \min_{\mathcal{Q}^{++}_\rho} \sum_{i=1}^N \int_{\mathbb{R}^d}\|\nabla \phi_i\|^2,
\end{equation}
where
\begin{equation}
    \mathcal{Q}^{++}_\rho =\{ (\phi_1, \phi_2, ..., \phi_N) \in H^1(\mathbb{R}^d; \mathbb{R}^N) \,:\, (\phi_i, \phi_j)_{L^2(\mathbb{R}^d; \mathbb{R})} =\delta_{ij}, 2\sum_{i=1}^N|\phi_i|^2 = \rho \}.
\end{equation}
The rest of the paper is concerned with the restricted non-interacting kinetic energy functional $T_{\text{RKS}}$. 

\subsection{Non-interacting $v$-representable densities}
Define a collection of orthonormal $N$ orbitals, for some $N \in \mathbb{N}$, as 
\begin{equation}
        \mathcal{P}_N = \{ \Phi = (\phi_1, \phi_2, ..., \phi_{N}) \in H^1(\mathbb{R}^d; \mathbb{R}^{N}) \,:\, (\phi_i, \phi_j)_{L^2(\mathbb{R}^d)} = \delta_{ij}, \, 1 \le i,j \le N \}.
\end{equation}
We single out a special collection of potentials $V_N$ as follows:
\begin{equation} \label{eq:Wn_def}
V_N = \left\{v \in Y^d : \exists\,\Phi^\circ \in \mathcal{P}_N\text{ s.t. } \Phi^\circ = \underset{\Phi \in \mathcal{P}_N}{\text{arg inf}} \; \sum_{i=1}^N \int_{\mathbb{R}^d} \left( \|\nabla \phi_i\|^2 + 2 \phi_i^2 v \right)\right\}.
\end{equation}
Thus, for each $v \in V_N$, we define
\begin{equation}
   \label{eq:non_interacting_minimization_problem}
    I^N(v) = \min_{\Phi \in \mathcal{P}_N} \sum_{i=1}^N \int_{\mathbb{R}^d} \left( \|\nabla \phi_i\|^2 + 2 \phi_i^2 v \right) = \sum_{i=1}^N \int_{\mathbb{R}^d} \left( \|\nabla \phi^\circ_i\|^2 + 2 (\phi^\circ_i)^2 v \right). 
\end{equation}
Following the original formulation given by ~\cite{KS65}, the Euler-Lagrange equations of the minimization problem \eqref{eq:non_interacting_minimization_problem} yields the following eigenvalue problem: for $1 \le k \le N$ and for every $w \in H^1(\mathbb{R}^d)$, 
\begin{equation} \label{eq:v_Wn_EL}
    \frac{1}{2}\int_{\mathbb{R}^d} \nabla \phi_k^\circ \cdot \nabla w +  \int_{\mathbb{R}^d} v \phi_k^\circ w = \sum_{i=1}^N A_{ki} \int_{\mathbb{R}^d} \phi_i w  
\end{equation}
The Lagrange multipliers $A_{ki}$, whose existence is known from Proposition~\ref{prop: lag_multiplier_on_banach}, can be arranged in the form of a real and symmetric matrix $A$ of size $N \times N$. Diagonalizing $A$ yields $A = P \Lambda P^T$, where $\Lambda$ is a diagonal matrix whose diagonal elements contain the eigenvalues $\lambda_i$ and $P$ is the orthogonal matrix whose columns are the eigenvectors of $A$. Define $\psi_i = \sum_{j=1}^N P_{ji}\phi_j^\circ$---it is easy to see that ${\psi_1, \psi_2, ..., \psi_N} \in \mathcal{P}_N$. We can also express the orbitals $(\phi_i^\circ)_{i=1}^N$ in terms of the eigenvectors of $A$ as $\phi_k^\circ = \sum_{j=1}^N P_{kj} \psi_j$. Substituting this into equation~\eqref{eq:v_Wn_EL} yields, for every $1 \le j \le N$,
\begin{equation*}
        \frac{1}{2}\int_{\mathbb{R}^d} \nabla \psi_j \cdot \nabla w + \int_{\mathbb{R}^d} v \psi_j w = \lambda_j \int_{\mathbb{R}^d} \psi_j w. 
\end{equation*}
Note that $(\psi_1, \psi_2, ... \psi_N)$ is also a minimizer of \eqref{eq:non_interacting_minimization_problem}. This implies that the infimum of the spectrum of $-\frac{1}{2}\Delta + v$ is achieved in the sense of distributions, and that $(\psi_i)_{i=1}^N$ constitute the $N$ eigenfunctions of $-\frac{1}{2}\Delta + v$. We call the minimizers $(\psi_i)_{i=1}^N$ as restricted Kohn-Sham orbitals. The orbital corresponding to the lowest eigenvalue will be referred to as the \emph{lowest eigenvalue} orbital. Note that if the spectrum is degenerate there can be multiple lowest eigenvalue orbitals. 

\begin{figure}[h]
    \centering
    \includegraphics[width=0.65\linewidth]{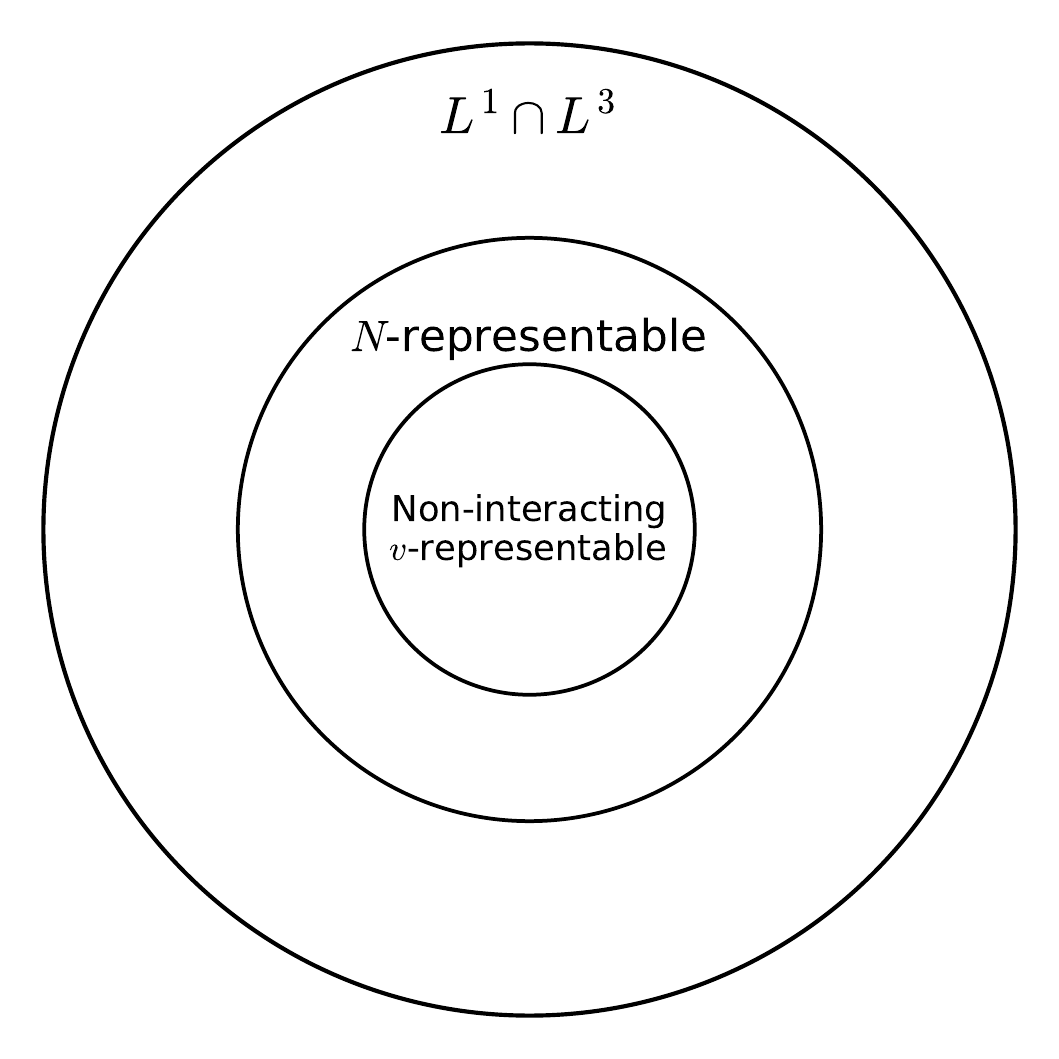}
    \caption{Schematic representation of the space of $N$-representable
    densities and the subset consisting of non-interacting
    $v$-representable densities.}
    \label{fig:density_space}
\end{figure}

The ground state density $\rho_v$ corresponding to a potential $v \in V_N$ is computed in terms of the restricted Kohn-Sham orbitals as $\rho_v= 2\sum_{i=1}^N \psi_i^2$. Note that $\rho_v$ is an $N$-representable density since $\rho_v \in \mathcal{I}_{2N}$. We call the collection of all such ground-state densities as \emph{non-interacting $v$-representable densities}. A precise characterization of which subset of $\mathcal{I}_N$ corresponds to non-interacting $v$-representable densities is unknown. The relation between these two classes of densities is illustrated schematically in Figure~\ref{fig:density_space}. The space of non-interacting $v$-representable densities is a strict subset of the
larger space of $N$-representable densities. Note that any ground state density computed using DFT is a non-interacting $v$-representable density. For future use, we note that the kinetic energy in this special case is computed as $T_{\text{RKS}}(\rho_v) = \sum_{i=1}^N \int_{\mathbb{R}^d} \|\nabla \psi_i\|^2$.

We also present a lemma that will be useful in the subsequent analysis:
\begin{lemma}
\label{lem: RKS_minimizer}
A lowest eigenvalue orbital of \( v \in V_N \) minimizes the energy functional
\[
\mathcal{E}(f) = \frac{1}{2} \int_{\mathbb{R}^d} \|\nabla f\|^2 + \int_{\mathbb{R}^d} v f^2
\]
over \( H^1(\mathbb{R}^d) \), subject to the normalization constraint \( \|f\|_{L^2(\mathbb{R}^d)} = 1 \). Furthermore, the modulus of any minimizer is also a minimizer and thus qualifies as a lowest eigenvalue orbital of \( v \). In addition, if  \( v \in L^\infty(\mathbb{R}^d) \), then the lowest eigenvalue orbital is unique up to a change of sign and can be chosen to be strictly positive.
\end{lemma}
\begin{proof}
    Let $\psi_1$ be a lowest eigenvalue RKS orbital of $v \in V_N$ and let $\lambda_1$ be the eigenvalue corresponding to $\psi_1$.
    Then the infimum of the spectrum of the operator \[
-\frac{1}{2} \Delta + v,
\]
in the sense of distributions, is attained by \( \psi_1 \). Since \( -\frac{1}{2} \Delta + v \) is self-adjoint, the min-max theorem \cite{RS78} implies that
\[
\frac{1}{2} \int_{\mathbb{R}^d} \lVert\nabla \psi_1\rVert^2 + \int_{\mathbb{R}^d} v \psi_1^2
\]
is the minimum of \( \mathcal{E}(f) \) over \( H^1(\mathbb{R}^d) \) under the constraint \( \|f\|_{L^2(\mathbb{R}^d)} = 1 \).
 Note that $\mathcal{E}(\psi_1) = \mathcal{E}(|\psi_1|) = \lambda_1$ and $|\psi_1|$ belong to the minimizing set. Therefore $|\psi_1|$ satisfies the following Euler-Lagrange equation:
    \begin{equation*}
          \frac{1}{2}\int_{\mathbb{R}^d} \nabla |\psi_1| \cdot \nabla w + \int_{\mathbb{R}^d} v |\psi_1| w = \lambda_1 \int_{\mathbb{R}^d} |\psi_1| w,  
    \end{equation*}
for all $w \in H^1(\mathbb{R}^d)$, thereby showing that $|\psi_1|$ is also the lowest eigenvalue orbital. The final claim follows from Proposition~\ref{prop: Lieb_Schrodinger_GS}.
\end{proof}

\section{A Variational Principle for computing the Non-Interacting Kinetic Energy}
\subsection{An informal introduction}
As shown earlier in equation~\eqref{eq:T_rks}, the non-interacting kinetic energy functional can be formulated as a minimization problem over orthonormal functions whose squared moduli sum to a given $N$-representable electron density. Our approach to computing the non-interacting kinetic energy centers on an exact trigonometric reparametrization of the density in terms of \emph{angle fields}. Informally, given a density $\rho$, we construct $N$ orbitals $(\phi_k^\Theta)_{k=1}^N$ in terms of $N-1$ functions $\Theta = (\theta_l:\mathbb{R}^d \to [-\pi/2,\pi/2])_{l=1}^{N-1}$, which we refer to as angle fields, as follows:
\begin{displaymath}
\begin{split}
\phi_N^\Theta &= \sqrt{\frac{\rho}{2}} \prod_{l=1}^{N-1} \cos \theta_l,\\
\phi_k^\Theta &= \sqrt{\frac{\rho}{2}} \sin \theta_k \prod_{l=1}^{k-1} \cos \theta_l, \quad k=1,\ldots,N-1.
\end{split}
\end{displaymath}

This parametrization addresses the following points:  
\begin{enumerate}[(i)]
\item The key constraint $2\sum_{k=1}^N (\phi_k^\Theta)^2 = \rho$
in the minimization problem~\eqref{eq:T_rks} is automatically satisfied. Without enforcing this constraint directly, one would need to introduce a Lagrange multiplier, 
leading to an adjoint problem and the possible complication of singular multipliers.\footnote{In fact, the adjoint problem is not always well-posed, since it requires the Fr\'echet derivative of the constraint to be surjective; see Proposition~\ref{prop: lag_multiplier_on_banach}.}

\item This approach applies to \emph{any} \(N\)-representable density and yields the energy explicitly, unlike Kohn--Sham DFT, which is restricted to non-interacting \(v\)-representable densities, whence there is no way of explicitly calculating energy for a given electron density.  
\end{enumerate}


The specific reason for restricting the range of the angle fields to $[-\pi/2,\pi/2]$ is to remove discontinuities that can arise in the minimizer---a simple example is provided in Appendix~\ref{app:discontinuity}. A more rigorous argument in support of this range restriction will be provided below.

In terms of the trigonometric parametrization just introduced, the minimization problem for computing $T_{\text{RKS}}(\rho)$, for a non-interacting $v$-representable density, $\rho$, can be equivalently written as a minimization problem involving the angle fields as follows:
\begin{displaymath}
\begin{split}
T_{\text{RKS}}(\rho) = \frac{1}{2} \inf_{\Theta} & \int_{\mathbb{R}^d} \|\nabla \sqrt{\rho}\|^2  + \int_{\mathbb{R}^d} \rho \|\nabla \theta_1\|^2  
    + \int_{\mathbb{R}^d} \rho\Big(  \cos^2{\theta_1}\|\nabla \theta_2\|^2 \\ & +  \cos^2{\theta_1}\cos^2{\theta_2}\|\nabla \theta_3\|^2 + ... + \cos^2{\theta_1}\cos{\theta_2}^2 ....\cos^2{\theta_{N-2}} \|\nabla \theta_{N-1}\|^2 \Big). \\
    \end{split}
\end{displaymath}
The specific set over which $\Theta$ is minimized will be precisely defined later in this section.

We now present a detailed analysis of the foregoing sketch. For purposes of mathematical analysis, we first extend the definition of the non-interacting kinetic energy to a much larger class of densities than $\mathcal{I}_{2N}$, specifically densities $\rho$ for which $\int_{\mathbb{R}^d} \rho$ is not an even integer. Note that this extension to $T_{\text{KS}}(\rho)$ remains well-defined even if intermediate densities obtained during an iterative algorithm are not $N$-representable.

\subsection{Extension of the non-interacting kinetic energy functional}
Recall that the functional $T_{\text{RKS}}(\rho)$ is defined for $N$-representable densities $\rho \in \mathcal{I}_N$ that integrate to an even integer. In the Kohn-Sham theory, it is computed on a much smaller subset of $v$-representable densities. As noted earlier, an explicit specification of the set of all $v$-representable densities is not available. In this work, we move in the opposite direction and generalize $T_{\text{RKS}}(\rho)$ to a much larger set than $\mathcal{I}_{2N}$. In particular, we extend it to include \emph{pseudo-densities} which integrate to a non-integer value.

Following standard practice, we develop the formalism in terms of the square root of the electron density $u = \sqrt{\rho}$. We associate with any $u \in L^2({\mathbb{R}^d})$ a  natural number $N_u$ as follows:
\begin{equation}
N_u = \text{ceil}\left(\int_{\mathbb{R}^d} \frac{u^2}{2}\right).
\end{equation}
Note that $N_u =0$ if and only if $u =0$ a.e. in $\mathbb{R}^d$.

Additionally, we associate with each $u \in H^1(\mathbb{R}^d)$, such that for $N_u \geq 1$, an \emph{almost-orthonormal} set $\mathcal{M}_u$ defined as
\begin{equation}
\begin{split}
        \mathcal{M}_u = \{ &(\phi_1, \phi_2, ..., \phi_{N_u}) \in H^1(\mathbb{R}^d; \mathbb{R}^{N_u}) : (\phi_i, \phi_j)_{L^2(\mathbb{R}^d)} = 0,\;1 \le i,j \le N_u, i \neq j, \\ &\|\phi_i\|_{L^2(\mathbb{R}^d)} = 1 \; \forall i \in [1, N_u-1], \; 2 \sum_{i=1}^{N_u}|\phi_i|^2 = u^2 \}.
\end{split}
\end{equation}
Note that $\lVert \phi_{N_u} \rVert^2_{L^2(\mathbb{R}^d)} = \int_{\mathbb{R}^d} u^2/2 - N_u + 1$, which is strictly positive and not equal to $1$ unless $\int_{\mathbb{R}^d} u^2/2$ is an integer. In preparation for the extension of the non-interacting kinetic energy functional, for all $u \in H^1(\mathbb{R}^d)$ such that $N_u \geq 1$, we define a functional $L: H^1(\mathbb{R}^d; \mathbb{R}^{N_u}) \rightarrow \mathbb{R}^+$ as
\begin{equation*}
       L(\phi_1, \phi_2, ..., \phi_{N_u}) = \sum_{i=1}^{N_u} \int_{\mathbb{R}^d}\|\nabla \phi_i\|^2.
\end{equation*}

\begin{lemma}
\label{lem:Lbounds}
    Let $u \in H^1(\mathbb{R}^d)$ such that $N_u \geq 1$. Then there exists $\Phi = (\phi_1, \ldots, \phi_{N_u}) \in \mathcal{M}_u$ such that $L(\Phi) \leq C \int_{\mathbb{R}^d} \|\nabla u\|^2$ for some positive constant $C$ that depends on $N_u$.  Moreover, for any $\Phi \in \mathcal{M}_u$, $$\frac{1}{2}\int_{\mathbb{R}^d} \|\nabla u\|^2 \leq L(\Phi).$$ 
\end{lemma}
\begin{proof}
We show first that $\mathcal{M}_u$ is non-empty and establish the upper bound on $L$. If $N_u = 1$, $\phi_1 = \frac{|u|}{\sqrt{2}}$ belongs to $\mathcal{M}_u$ and upper bound on $L$ is trivial. Suppose then that $N_u \geq 2$, and define $z = N_u-1$ and $\alpha = 2z/\int_{\mathbb{R}^d} u^2$. Noting that \( 1/2 \leq \alpha < 1 \), we define the function \( g: \mathbb{R} \to \mathbb{R} \) as follows---see Figure~\ref{fig:g_trapz}:

\[
g(x) =
\begin{cases}
\sqrt{\dfrac{1 + \alpha}{2}} \cdot \dfrac{4(1 + \alpha)}{3(1 - \alpha)} \cdot x & \text{if } 0 \le x < \dfrac{3(1 - \alpha)}{4(1 + \alpha)} \\[2ex]
\sqrt{\dfrac{1 + \alpha}{2}} & \text{if } \dfrac{3(1 - \alpha)}{4(1 + \alpha)} \le x \le 1 - \dfrac{3(1 - \alpha)}{4(1 + \alpha)} \\[2ex]
\sqrt{\dfrac{1 + \alpha}{2}} \cdot \dfrac{4(1 + \alpha)}{3(1 - \alpha)} \cdot (1 - x) & \text{if } 1 - \dfrac{3(1 - \alpha)}{4(1 + \alpha)} < x \le 1 \\[2ex]
0 & \text{otherwise}
\end{cases}
\]

\begin{figure}[h]
    \centering  \includegraphics[width=0.5\textwidth]{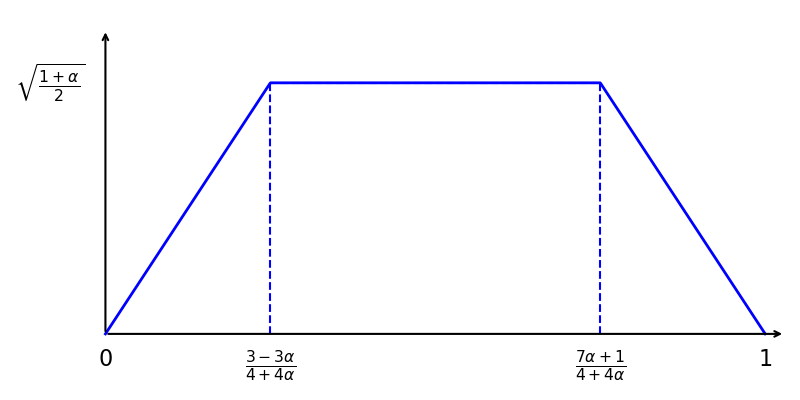}
    \caption{Plot of $g(x)$ over $[0,1]$.}
    \label{fig:g_trapz}
\end{figure}

It is easy to see that $g \in H^1(\mathbb{R})$, $g^2 \leq \frac{1+\alpha}{2} $, $\int_0^1 g^2 = \alpha $ and $g'^2$ is bounded by a constant.  
Borrowing ideas from \cite{MH58} and \cite{Lieb83}, define, for $x = (x^1, \ldots, x^d) \in \mathbb{R}^d$, 
    \begin{equation*}
        f(x^1) = \frac{2z \int_{-\infty}^{x^1} ds \int_{\mathbb{R}^{d-1}}dt \quad u^2(s, t)}{{\int_{\mathbb{R}^d} u^2}} -z. 
    \end{equation*}
Note that $f$ is monotonically increasing from $-z$ to $z$. Defining 
\begin{equation}
\label{eq:trial_phi}
\begin{split}
    \phi^k(x) &= \frac{|u(x)|}{\sqrt{2}} (g(f(x^1) - k+1) - g(f(x^1) +k)) , \quad k = 1, \ldots, N_u - 1,\\
    \phi^{N_u}(x) &= \sqrt{\frac{u(x)^2}{2} - \sum_{k=1}^{N_u-1}|\phi^k(x)|^2},
\end{split}
\end{equation}
it follows that $\{\phi^k\}_{k=1}^{N_u-1}$ is an orthonormal set in $L^2(\mathbb{R}^d; \mathbb{R})$, $|\phi^k|^2 \leq \frac{1+\alpha}{2} \frac{u^2}{2}$ for all $1 \leq k \leq N_u - 1$, and $\sum_{k=1}^{N_u-1}|\phi^k|^2 \leq \frac{1+\alpha}{2} \frac{u^2}{2} $.  Noting that $\phi^k/|u|$ is an odd function with respect to $f(x^1)$ whereas $\phi^{N_u}/|u| $ is even, it is evident that $\phi^{N_u}$ is orthogonal to all the elements of the above mentioned orthonormal set. Hence, $\Phi = (\phi^1, \phi^2, ..., \phi^{N_u}) \in \mathcal{M}_u$, thereby establishing that $\mathcal{M}_u$ is non-empty.

It is also easy to see that for $1 \leq k \leq N_u-1$,
\begin{equation*} 2\|\nabla \phi^k(x)\|^2 \leq \| \nabla u \|^2 + u(x)^2 (g'(f(x^1) - k+1) - g'(f(x^1) +k))^2 f'(x^1)^2. 
\end{equation*}
Since $g'$ is bounded by a constant, we get, following the proof in Theorem 1.2 of \cite{Lieb83}, that $$\sum_{k=1}^{N_u-1} \int_{\mathbb{R}^d}\|\nabla \phi^k(x)\|^2 \leq C_1 N_u^2 \int_{\mathbb{R}^d} \| \nabla u \|^2,$$
where $C_1 \geq 1$ is independent of $u$ and $d$. We also have $$2\phi^{N_u}\nabla \phi^{N_u} = u\nabla u - 2\sum_{k=1}^{N_u-1}\phi^k \nabla \phi^k .$$ 
Note that $\frac{(1-\alpha)|u|}{4} \leq \phi^{N_u} \leq \frac{|u|}{2}$ and recall that $1/2 \leq \alpha < 1$. Using the Cauchy-Schwarz inequality, we get the following estimate:
\begin{equation*}
    \| \nabla \phi^{N_u} \|^2 \leq C_2 \left( \|\nabla u\|^2 + \sum_{k=1}^{N_u-1} \| \nabla \phi^k\|^2 \right)
\end{equation*}
where $C_2$ is independent of $u$ and $d$, thus concluding the first part of the proof. 

To prove the second part, note that for any $\Phi = (\phi_1, \phi_2, ... ,\phi_{N_u}) \in \mathcal{M}_u$, we have $\sum_{i=1}^{N_u}\phi_i^2 = \frac{u^2}{2}$. Taking the weak derivative on both sides and using the Cauchy-Schwarz inequality proves the second part of the proof.
\end{proof}

\begin{remark}
The upper bound in the foregoing lemma is not optimal with respect to $N_u$. However, optimal bounds were obtained in \cite{ZM83} and \cite{GB96}.
\end{remark}

\begin{remark}
Note that $\phi^{N_u} \ge 0$. This observation will be used later in a later section. 
\end{remark}

The foregoing lemma shows that $\mathcal{M}_u$ is non-empty, and that $L$ is bounded from below. We can hence define $T : H^1(\mathbb{R}^d) \rightarrow \mathbb{R}_{\geq 0}$ such that
\begin{equation}
\label{eq:firstT}
T(u) = 
\begin{cases}
0, & N_u = 0 , \\
\inf_{\Phi \in \mathcal{M}_u} L(\Phi), & \text{$N_u \geq 1$}.
\end{cases}
\end{equation}

\begin{theorem}
    \label{thm:min_exist}
    The infimum in equation~\eqref{eq:firstT} is achieved for any $u \in H^1(\mathbb{R}^d)$ with $N_u \geq 1$. 
\end{theorem}
\begin{proof}
Let $\Phi_i = (\phi^1_i, \phi^2_i, ..., \phi^{N_u}_i)$ be a minimizing sequence for $\inf_{\Phi \in \mathcal{M}_u} L(\Phi)$. It follows from Lemma~\ref{lem:Lbounds} that the sequences $\{\phi^k_i\}_{i=1}^\infty$ are bounded in $H^1(\mathbb{R}^d)$ for all $1 \leq k \leq N_u$. Hence $\phi^k_i$ converge along a common subsequence to some $\phi^k$ weakly in $H^1(\mathbb{R}^d)$. We get, using Theorem 3.3 in \cite{Lieb83}, strong convergence in $L^2(\mathbb{R}^d)$ 
and hence pointwise almost everywhere in $\mathbb{R}^d$ along a common subsequence, for all $1 \leq k \leq N_u$. For the sake of convenience we notate this subsequence with same subscript $i$. Using pointwise convergence, we see that $2\sum_{k=1}^{N_u}(\phi^k)^2 = u^2$. Similarly, using dominated convergence theorem, we can establish that $(\phi^1, \phi^2, .., \phi^{N_u})$ is almost orthonormal since $|\phi^k_i| \leq |u|$ for all $1 \leq k \leq N_u$. It follows that $\Phi= (\phi^1, \phi^2, .., \phi^{N_u}) \in \mathcal{M}_u$. Using the weak limit we also have $\int_{\mathbb{R}^d} \|\nabla \phi^k\|^2  \leq \liminf_{i \rightarrow \infty} \int_{\mathbb{R}^d} \|\nabla \phi_i^k\|^2 $, which concludes the proof. 
\end{proof}
This theorem show that there exists a minimimizer $\Phi \in \mathcal{M}_u$ of $L$. We note that this is trivially true for all $u$ such that $N_u = 1$, in which case, $T(u) = \frac{1}{2}\int_{\mathbb{R}^d} \|\nabla u\|^2 $. Also note that $T(\sqrt{\rho})=T_{RKS}(\rho)$ for every $\rho \in \mathcal{I}_{2N}$. This shows that $T$ extends $T_{\text{RKS}}$ for all densities with $\sqrt{\rho} \in H^1(\mathbb{R}^d)$. We work with $T$ instead of $T_{\text{RKS}}$ in what follows.

\subsection{Minimization problem} \label{sec:minprob}
Let $u \in H^1(\mathbb{R}^d)$ such that $N_u \geq 2$. For some $N \in \mathbb{N}$, let $\Theta = (\theta_1, \theta_2, ..., \theta_{N-1}) \in L^1_{loc}(\mathbb{R}^d; \mathbb{R}^{N-1})$. Then, we first define following $N-1$ functions:
\begin{equation*}
    u_k = \frac{|u|}{\sqrt{2}} \prod_{i=1}^{k-1} \cos{\theta_i}, \quad 1 \leq k < N.
\end{equation*}
Further, we construct a set of orbitals $(\varphi_k^{\Theta, u})_{k=1}^N$ as follows:
\begin{equation}
\begin{split}
\varphi^{\Theta, u}_N &= u_{N-1} \cos{\theta_{N-1}} = \frac{|u|}{\sqrt{2}} \prod_{i=1}^{N-1} \cos{\theta_i},\\
\varphi_k^{\Theta, u} &= u_k \sin \theta_k = \frac{|u|}{\sqrt{2}} \sin \theta_k \prod_{i=1}^{k-1} \cos{\theta_i}, \quad 1 \leq k\leq N-1.
\end{split}
\end{equation}
We will use the shorthand notation $\Phi^{\Theta, u} = (\varphi^{\Theta, u}_1, \varphi^{\Theta, u}_2, ...., \varphi^{\Theta, u}_N)$ wherever appropriate. As noted in the introductory remarks, the advantage with this parameterization is that these orbitals satisfy the following constraint exactly:
$$ 2\sum_{i=1}^N |\varphi^{\Theta, u}_i|^2 = u^2.$$
We define a collection of angle fields as follows:
\begin{equation}
\begin{split}
\mathcal{S}_u = \{(\theta_1, \theta_2, ..., \theta_{N_u-1}) \in H^1_{loc}(\mathbb{R}^d; \mathbb{R}^{N_u-1})\;:\; u_k \nabla \theta_k \in L^2(\mathbb{R}^d; \mathbb{R}^3), \; 1 \leq k\leq N_u-1 \}. 
\end{split}
\end{equation}
It follows that $\Theta \in \mathcal{S}_u$ implies $\varphi^{\Theta, u}_k$ belong to $H^1(\mathbb{R}^d)$ for all $1 \leq k \leq N_u$. Further, we define $\Tilde{L}: \mathcal{S}_u  \rightarrow \mathbb{R}^+$ as follows:
$$\Tilde{L}(\Theta) = L(\Phi^{\Theta, u}) = \sum_{k=1}^{N_u} \int_{\mathbb{R}^d}\|\nabla \varphi_k^{\Theta, u}\|^2.$$ 
Using induction it follows that
\begin{displaymath}
    \begin{split}
    \Tilde{L}(\Theta) &= \frac{1}{2} \int_{\mathbb{R}^d} \|\nabla u\|^2  + \frac{1}{2} \int_{\mathbb{R}^d} u^2 \|\nabla \theta_1\|^2  
    + \frac{1}{2} \int_{\mathbb{R}^d} u^2 \Big(  \cos^2{\theta_1}\|\nabla \theta_2\|^2 \\ & +  \cos^2{\theta_1}\cos^2{\theta_2}\|\nabla \theta_3\|^2 + ... + \cos^2{\theta_1}\cos{\theta_2}^2 ....\cos^2{\theta_{N_u-2}} \| \nabla \theta_{N_u-1}\|^2 \Big). \\
    \end{split}
\end{displaymath}
We define next a collection of angle fields that correspond to almost orthonormal orbitals:
\begin{displaymath}
    \begin{split}
        \mathcal{A}_u = \{ \Theta \in L^1_{loc}(\mathbb{R}^d; \mathbb{R}^{N_u-1}) \;:\; &\langle \varphi^{\Theta, u}_i, \varphi^{\Theta, u}_j \rangle_{L^2(R^d, \mathbb{R})} = 0, \, 1 \le i < j \le N_u, \\& \lVert \varphi^{\Theta, u}_i\rVert_{L^2(\mathbb{R}^d)} = 1,\, 1 \le i < N_u\}.
    \end{split}
\end{displaymath}

\begin{lemma}
    \label{lem:upper_bound}
    Let $u \in H^1(\mathbb{R}^d)$ with $N_u \geq 2$. Then $\mathcal{S}_{u} \cap \mathcal{A}_{u}$ is non-empty.
\end{lemma}
\begin{proof}
    For all $1 \leq j < N_u$ define 
    \begin{equation*}
        \theta_j = \arcsin{(g(f(x^1) - k+1) - g(f(x^1) +k))}
    \end{equation*}
using Lemma~\ref{lem:Lbounds}. Note that $supp(\theta_i) \cap supp(\theta_j)$ has zero measure when $i \neq j$. It is then easy to see that $\Theta \in \mathcal{S}_u$ since $0 < \cos{\sqrt{\frac{1+\alpha}{2}}}\cos{\theta_i} \leq 1$. Further note that $\phi^{\Theta, u}_k = \phi_k$, with $\phi_k$'s defined as in Lemma~\ref{lem:Lbounds}. Hence $\Theta \in \mathcal{A}_u$.
\end{proof}

This lemma establishes that for every $\Theta \in \mathcal{S}_{u} \cap \mathcal{A}_{u}$, 
there exists a corresponding $\Phi^{\Theta,u} \in \mathcal{M}_u$. However, it is not necessarily true that for every $\Phi \in \mathcal{M}_u$ there exists a $\Theta \in \mathcal{S}_{u} \cap \mathcal{A}_{u}$ such that $\Phi^{\Theta, u} = \Phi$.
Attempting to define $\Theta$ pointwise in order to satisfy this equality 
leads to two difficulties. 
First, $\arcsin(\phi/u)$ can be discontinuous 
even when $\phi$ itself is continuous. 
Second, using the $\mathrm{atan2}$ function to determine 
$\theta_{N_u-1}$ from the ratio of 
$\phi^{\Theta,u}_{N_u-1}$ and $\phi^{\Theta,u}_{N_u}$ will also yield discontinuity. 
Explicit counterexamples for both situations are provided in 
Appendix~\ref{app:discontinuity}.
This implies that we cannot replace minimization over $\Phi$ with one over $\Phi^{\Theta, u}$. To resolve this issue, we start by defining
\begin{equation}
        \mathcal{M}_u^+ = \{ \Phi = (\phi_1, \phi_2, ..., \phi_{N_u}) \in \mathcal{M}_u \;:\; \phi_{N_u} \geq 0 \}.
\end{equation}
As remarked in Lemma~\ref{lem:Lbounds}, the set $\mathcal{M}_u^+$ is also non-empty. Following the same reasoning as in Theorem~\ref{thm:min_exist} we define 
\begin{equation} \label{eq:Tplusmin}
    T^+ (u) = \min_{\Phi \in \mathcal{M}_u^+} L(\Phi).
\end{equation}
Since $\mathcal{M}_u^+ \subset \mathcal{M}_u$, it holds that $T(u) \leq T^+(u)$. Further, using Lemma~\ref{lem: RKS_minimizer} it follows that $T(u) = T^+(u)$ when $u^2$ is a non-interacting $v$-representable ground state density. We resolve the issues raised earlier in this section by recourse to the following theorem.

\begin{theorem} \label{thm:energy_derivation}  Let $u \in H^1(\mathbb{R}^d)$ with $N_u \geq 2$ and define a subset of $\mathcal{S}_{u}$ as follows:
\begin{equation*}
    \mathcal{S}_{u}^+ = \{ \Theta \in \mathcal{S}_{u} \;:\; |\theta_i| \le \pi/2, \; 1 \le i < N_u \}. 
\end{equation*}
Then, 
\begin{equation} \label{eq:Tplusinfnew}
    T^+(u) = \lim_{\delta \downarrow 0^+}\min_{\Theta \in \mathcal{S}_{u}^+ \cap   \mathcal{A}^\delta_u } \Tilde{L}(\Theta),
\end{equation}
where, for every $\delta >0$, 
\begin{displaymath}
    \begin{split}
           \mathcal{A}^\delta_u = \{\Theta \in L^1_{loc}(\mathbb{R}^d; \mathbb{R}^{N_u-1}) \;:\; &|\langle \varphi^{\Theta, u}_i, \varphi^{\Theta, u}_j \rangle_{L^2(R^d, \mathbb{R})}| \leq \delta, \, 1 \le i < j \le N_u, \\ & |\lVert \varphi^{\Theta, u}_i\rVert_{L^2(\mathbb{R}^d)} -1| \leq  \delta,\, 1 \le i < N_u\}. 
    \end{split}
\end{displaymath}
\end{theorem}

\begin{proof}
The proof is split into two parts. In the first part, we show that there exists a sequence $\{\Theta_{\epsilon}\}$ in $\mathcal{S}_u^+$ such that $\Tilde{L}(\Theta_{\epsilon})$ converges to $T^+(u)$. Later we show that a subsequence of the original sequence belongs to $\mathcal{A}_u^\delta$ and conclude the proof of the theorem.

We begin by noting that since $u \neq 0$ a.e., $\sup |u| > 0$. However, $\inf |u| = 0$ because $u \in L^2(\mathbb{R}^d)$. This can cause trouble for weak differentiability when dealing with division by $u$. To resolve this, we let $\Phi = (\phi_1, \phi_2, ..., \phi_{N_u}) \in \mathcal{M}_u^+$ be such that 
 $T^+(u) = L(\Phi)$. For $1 > \epsilon_1 > 0$  define $u_{1, \epsilon_1} = \frac{|u|}{\sqrt{2}}$ and
\begin{equation*}
    h_{1, \epsilon_1} =\frac{\phi_1}{u_{1, \epsilon_1} + \epsilon_1}.
\end{equation*}
Recall that $\phi_1^2 \leq u_{1,\epsilon_1}^2$, and $\phi_1 \in H^1(\mathbb{R}^d)$. Therefore, $h_{1,\epsilon_1}$ is strictly bounded by $1$ and belongs to $H^1_{loc}(\mathbb{R}^d)$. Hence, defining $\theta_{1,\epsilon_1} = \sin^{-1}(h_{1,\epsilon_1})$, we see that $|\theta_{1,\epsilon_1}| < \frac{\pi}{2}$ and $\theta_{1,\epsilon_1}$ is weakly differentiable. We recursively define, for $1 \leq k \leq N_u-1$ and $1 > \epsilon_k > 0$, 
\begin{equation}
\label{eq:mainthm_theta_def}
u_{k, \epsilon_k} = \frac{|u|}{\sqrt{2}} \prod_{i=1}^{k-1} \cos{\theta_{i, \epsilon_i}}, \quad h_{k, \epsilon_k} = \frac{\phi_k} {u_{k,\epsilon_{k}} + \epsilon_k}, \quad \theta_{k, \epsilon_k} = \sin^{-1}(h_{k,\epsilon_k}).
\end{equation}
For $1 \leq k \leq N_u-1$ we have
$$2\phi_k^2 \leq u^2 - 2\sum_{i=1}^{k-1}\phi_{i}^2,\quad\text{and}\quad u_{1, \epsilon_1}^2 - \sum_{i=1}^{k-1}u^2_{i,\epsilon_i}\sin^2{\theta_{i,\epsilon_i}} = u_{k, \epsilon_k}^2.$$ Using recursion we get $\phi_{k}^2 \leq u_{k,\epsilon_k}^2$, whence $|h_{k, \epsilon_k}| < 1$. Further, for $1 \leq k \leq N_u-1$, we have $u_{k, \epsilon_k} \geq 0$ and $u_{k, \epsilon_k} = 0$ only when $u=0$. Note that the inequalities shown above make all the functions well-defined. Similar to the case of $\theta_{1, \epsilon_1}$, it is easy to see that $|\theta_{k, \epsilon_k}| < \pi/2$ and $\theta_{k, \epsilon_k}$ is weakly differentiable for $1 \leq k \leq N_u-1$. 

Since $\sin{\theta_{k, \epsilon_k}} = h_{k, \epsilon_k}$ and $\cos{\theta_{k, \epsilon_k}} = \sqrt{1-h_{k, \epsilon_k}^2}$ for $1 \leq k \leq N_u-1$, we have $\cos{\theta_{k, \epsilon_k}} \nabla \theta_{k, \epsilon_k} = \nabla h_{k, \epsilon_k}$, and $-\sin{\theta_{k, \epsilon_k}} \nabla \theta_{k, \epsilon_k} = \nabla \sqrt{1-h_{k, \epsilon_k}^2}$, and hence that $$\left\lVert \nabla \theta_{k, \epsilon_k} \right \rVert^2 = \left\lVert\nabla h_{k, \epsilon_k}\right \rVert^2 + \left\lVert\nabla \sqrt{1-h_{k, \epsilon_k}^2}\right \rVert^2.$$

Noting that $u_{1, \epsilon_1} \in H^1(\mathbb{R}^d)$ we conclude that $\theta_{1, \epsilon_1} \in H^1_{loc}(\mathbb{R}^d)$ and $u_{1, \epsilon_1}\theta_{1, \epsilon_1} \in H^1(\mathbb{R}^d)$. Since $u_{k, \epsilon_k}=u_{k-1, \epsilon_{k-1}}\cos{\theta_{k-1, \epsilon_{k-1}}}$ we get, using a recursive argument, that both $u_{k, \epsilon_k}$ and $u_{k, \epsilon_k}\theta_{k, \epsilon_k}$ belongs to  $H^1(\mathbb{R}^d)$ for all $ 1 \leq k \leq N_u-1$. Moreover, $\theta_{k, \epsilon_k} \in H^1_{loc}(\mathbb{R}^d)$ for all $ 1 \leq k \leq N_u-1$. 

Let $\Theta_\epsilon = (\theta_{1, \epsilon_1}, \theta_{2, \epsilon_2}, \ldots, \theta_{N_u-2, \epsilon_{N_u-2}}, \theta_{N_u-1, \epsilon_{N_u-1}} )$ and note that $\Theta_\epsilon \in \mathcal{S}_{u}^+$. Then construct the corresponding set of orbitals $\Phi^{\Theta_\epsilon, u} = (\varphi_1^{\Theta_\epsilon, u}, \varphi_2^{\Theta_\epsilon, u}, ..., \varphi_{N_u}^{\Theta_\epsilon, u})$. It is easy to see that $\varphi_k^{\Theta_\epsilon, u} = u_{k, \epsilon_k} \sin{\theta_{k, \epsilon_k}}$ for $1 \leq k \leq N_u-1$. We thus have
\begin{equation}    \label{eq:mainthm_varphi_phi_rel}   \varphi_k^{\Theta_\epsilon, u} = \phi_k -  \epsilon_k \sin{\theta_{k, \epsilon_k}}, \quad 1 \leq k \leq N_u-1.
\end{equation}
It follows that $\varphi_k^{\Theta_\epsilon, u}$ converges to $\phi_k$ pointwise almost everywhere and thus $u_{k, \epsilon_k}$ converges to $\sqrt{u^2/2 - \sum_{i=1}^{k-1}\phi_i^2}$ pointwise almost everywhere for $1 \leq k \leq N_u-1$; note that $u_{1, \epsilon_1} = \frac{|u|}{\sqrt{2}}$. Further, using definition in equation~\eqref{eq:mainthm_theta_def}, we also have   
\begin{equation}
\label{eq:mainthm_theta_def_2}
\varphi_k^{\Theta_\epsilon, u} = \left(\frac{u_{k, \epsilon_k}}{u_{k, \epsilon_k} + \epsilon_k}\right) \phi_k,
\end{equation}
and thus, 
\begin{equation*}
    \nabla \varphi_k^{\Theta_\epsilon, u} = \phi_k \nabla \left(\frac{u_{k, \epsilon_k}}{u_{k, \epsilon_k} + \epsilon_k}\right)  + \left(\frac{u_{k, \epsilon_k}}{u_{k, \epsilon_k} +\epsilon_k}\right) \nabla \phi_k, \quad 1 \leq k \leq N_u-1.
\end{equation*}
 
 We introduce the following compact notation: $$\lim_{\epsilon \rightarrow 0^+} = \lim_{\epsilon_1 \rightarrow 0^+}\lim_{\epsilon_2 \rightarrow 0^+}...\lim_{\epsilon_{N_u-2} \rightarrow 0^+}\lim_{\epsilon_{N_u-1} \rightarrow 0^+}.$$ If it holds that, for $1 \leq k \leq N_u-1$, $$\lim_{\epsilon \rightarrow 0^+} \int_{\mathbb{R}^d} \phi_k^2 \left\lVert\nabla \left(\frac{u_{k, \epsilon_k}} {u_{k, \epsilon_k} + \epsilon_k}\right)\right\rVert^2 = 0,$$ we then get 
\begin{equation}
\label{eq:energy_equivalence}
    \lim_{\epsilon \rightarrow 0^+} \int_{\mathbb{R}^d} \left\lVert \nabla \varphi_k^{\Theta_\epsilon, u} \right\rVert^2 =  \int_{\mathbb{R}^d} \left\lVert \nabla \phi_k \right\rVert^2,
\end{equation}
using the Cauchy-Schwarz and the Dominated Convergence Theorem (DCT). Noting that
\begin{equation*}
    \phi_k^2 \left\lVert \nabla \left(\frac{u_{k, \epsilon_k}} {u_{k, \epsilon_k + \epsilon_k}}\right) \right\rVert^2 = \phi_k^2 \epsilon_k^2 \frac{\|\nabla u_{k, \epsilon_k}\|^2}{(u_{k, \epsilon_k} + \epsilon_k)^4}, 
\end{equation*}
and using the facts that $\phi_k^2 \leq (u_{k, \epsilon_k} + \epsilon_k)^2$ and $u_{k, \epsilon_k} \in H^1(\mathbb{R}^d)$, while not depending on $\epsilon_{k}$, we get using DCT that 
$$\lim_{\epsilon \rightarrow 0^+} \int_{\mathbb{R}^d} \phi_k^2 \left \lVert\nabla \frac{u_{k, \epsilon_k}} {u_{k, \epsilon_k} + \epsilon_k}\right\rVert^2 =0, \quad 1 \leq k \leq N_u-1.$$
Similarly, since 
$(\varphi^{\Theta_\epsilon, u}_{N_u})^2 =  u_{1, \epsilon_1}^2 - \sum_{i=1}^{N_u-1} (\varphi^{\Theta_\epsilon, u}_i)^2 $, using equation~\eqref{eq:mainthm_theta_def_2} and Proposition~\ref{prop:H1_level_set} it is easy to see that
\begin{equation*}
        \lim_{\epsilon \rightarrow 0^+}\int_{\mathbb{R}^d} \| \nabla \varphi^{\Theta_\epsilon, u}_{N_u} \|^2 = \int_{\mathbb{R}^d} \left \lVert \nabla \sqrt{u_{1, \epsilon_1}^2 - \sum_{i=1}^{N_u-1} \phi_i^2 } \right\rVert^2 . 
\end{equation*}

Using 
\begin{equation}    \label{eq:positivity_requirement}
     \phi_{N_u} =\sqrt{ u_{1, \epsilon_1}^2 - \sum_{i=1}^{N_u-1}\phi_i^2},
\end{equation} 
we get 
\begin{equation*}
            \lim_{\epsilon \rightarrow 0^+}\int_{\mathbb{R}^d} \| \nabla \varphi^{\Theta_\epsilon, u}_{N_u} \|^2 =  \int_{\mathbb{R}^d} \| \nabla \phi_{N_u} \|^2.
\end{equation*}
Further, using equation~\eqref{eq:energy_equivalence} we obtain
\begin{displaymath}
    \begin{split}
        \lim_{\epsilon \rightarrow 0^+} \Tilde{L}(\Theta_{\epsilon}) = T^+(u).
    \end{split}
\end{displaymath}
This concludes the first part of the proof. 

The second part of the proof focuses on the orthonormality relation. Since, $|\varphi_i^{\Theta_{\epsilon}, u}\varphi_j^{\Theta_{\epsilon}, u}| < u^2/2$, we get using DCT that
\begin{equation*}
 \lim_{\epsilon \rightarrow 0^+} \int_{\mathbb{R}^d} \varphi_i^{\Theta_{\epsilon}, u}\varphi_j^{\Theta_{\epsilon}, u} = \int_{\mathbb{R}^d} \phi_i\phi_j, \quad 1 \leq i, j \leq N.   
\end{equation*}
This implies that for every $\delta > 0$ there exists a sequence $\{\Theta_{\epsilon_j}\}_{j=1}^\infty\subset \mathcal{S}_{u}^+ \cap \mathcal{A}^\delta_u$ such that 
\begin{equation*}
        \lim_{j \rightarrow \infty} \Tilde{L}(\Theta_{\epsilon_j}) = T^+(u).
\end{equation*}
We therefore get 
\begin{equation}
\label{eq:delta_convg}
    \inf_{\Theta \in \mathcal{S}_u^+ \cap \mathcal{A}^\delta_u} \Tilde{L}(\Theta) \leq  T^+(u).
\end{equation}

Further, it is easy to see that there exists $\Theta_\delta \in \mathcal{S}_{u}^+ \cap \mathcal{A}^\delta_u$ such that $\varphi^{\Theta_\delta, u}_k$ is bounded in $H^1(\mathbb{R}^d)$ for all $1 \leq k \leq N_u$. Therefore, as $\delta \downarrow 0^+$, $\varphi^{\Theta_\delta, u}_k$ weakly convergence along a common subsequence in $H^1(\mathbb{R}^d)$ for all $1 \leq k \leq N_u$. Use the same argument as in Theorem~\ref{thm:min_exist} to see that $ T^+(u) \leq  \lim_{\delta \downarrow 0^+} \inf_{\Theta \in \mathcal{S}_{u}^+\cap \mathcal{A}^\delta_u } \Tilde{L}(\Theta)$.

Hence,
\begin{equation*}
    \lim_{\delta \downarrow 0^+} \inf_{\Theta \in \mathcal{S}_{u}^+\cap \mathcal{A}^\delta_u } \Tilde{L}(\Theta) = T^+(u), 
\end{equation*}
thereby concluding the proof.
\end{proof}

It is worth reiterating that $\mathcal{M}_u^+$ and $\mathcal{A}_u^+$ are introduced to avoid discontinuities of $\theta_{N_u}$ like the one shown in second part of Appendix~\ref{app:discontinuity}. Note that equation~\eqref{eq:positivity_requirement} would not be valid if we consider $\mathcal{M}_u$ instead of $\mathcal{M}_u^+$. Furthermore, to preserve the regularity of the
angle fields, we approximate the orbitals. This results in an infimum problem rather than a true minimum and also causes a loss of exact orthonormality---we therefore work with the approximate orthonormality condition.

Since $T(u)= T^+(u)$ when $u^2$ is a non-interacting $v$-representable density and $T^+$ can be computed using the minimization problem~\eqref{eq:Tplusinfnew}, we use $T^+$ as a practical definition of the non-interacting kinetic energy functional. 

\begin{remark}
    In practice, $u \in H^1_0\cap L^\infty(\Omega)$, where $\Omega \subseteq \mathbb{R}^d$ is open and bounded. Let $\Omega'$ be an open and bounded superset of $\Omega$ such that $dist(\Omega, \partial \Omega') > 0$. Using the fact that $C_c^\infty(\Omega')$ is dense in $H^1_0(\Omega')$ we get
    \begin{equation*}
            T^+(u) = \lim_{\delta \downarrow 0^+} \inf\{  \Tilde{L}(\Theta) : \Theta \in  C_c^\infty(\Omega', [-\pi/2, \pi/2]^{N_u-1})\cap \mathcal{A}^\delta_u\}.
    \end{equation*}
    This observation become particularly important for the numerical approximations presented in Section~\ref{sec:numerical_examples}, since this permits us to approximate the angle fields using a smooth basis.
\end{remark}

\section{An Analytical Density-to-Potential Construction in One Dimension}
\label{sec:analytical_example}

We next illustrate the advantage of the angle-field formulation by
constructing explicitly the density-to-potential map for a class of
one-dimensional four-electron systems. The construction proceeds by
first obtaining the cumulative-density angle field from the prescribed
density, then constructing the corresponding orthonormal orbitals, and
finally determining the common local potential from the Euler--Lagrange
equations.

\subsection{A class of densities}

Consider the family
\begin{equation}
    \rho_{A,c}(x)
    =
    \frac{4k}{\pi}\operatorname{sech}(kx+c),
    \qquad
    k=\sqrt{\frac{A\pi}{2}},
    \qquad
    A>0,\quad c\in\mathbb{R}.
    \label{eq:analytical_density}
\end{equation}
The density is strictly positive and bounded, decays to zero as
$x\to\pm\infty$, and satisfies
\begin{equation}
    \int_{-\infty}^{\infty}\rho_{A,c}(x)\,dx=4.
\end{equation}
Thus, it corresponds to a four-electron closed-shell system. The
parameter $A$ determines the spatial scale of the density, while $c$
translates the profile without changing its shape, as illustrated in
Figure~\ref{fig:analytical_density_family}. We denote its square
root by
\begin{equation}
    u_{A,c}(x)=\sqrt{\rho_{A,c}(x)}
    =
    \sqrt{\frac{4k}{\pi}}\,
    \operatorname{sech}^{1/2}(kx+c).
    \label{eq:analytical_u}
\end{equation}

\begin{figure}
    \centering
    \includegraphics[width=0.82\linewidth]{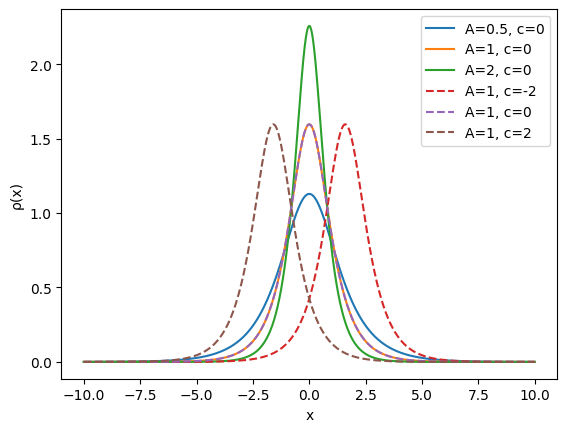}
    \caption{The family of densities $\rho_{A,c}$ for different values
    of the scale parameter $A$ and translation parameter $c$.}
    \label{fig:analytical_density_family}
\end{figure}

\subsection{The angle-field minimization}

For a prescribed $u=u_{A,c}$, the four-electron angle-field
parametrization reduces the kinetic-energy minimization to
\begin{equation}
    \min_{\theta}
    \frac{1}{2}
    \int_{\mathbb{R}}u^2(x)\theta'(x)^2\,dx,
    \label{eq:analytical_theta_min}
\end{equation}
subject to the orthonormality constraints
\begin{equation}
    \int_{\mathbb{R}}
    u^2(x)\sin\bigl(2\theta(x)\bigr)\,dx=0,
    \label{eq:analytical_constraint_sin}
\end{equation}
and
\begin{equation}
    \int_{\mathbb{R}}
    u^2(x)\cos\bigl(2\theta(x)\bigr)\,dx=0.
    \label{eq:analytical_constraint_cos}
\end{equation}
These constraints ensure that the two orbitals generated by the
angle field are orthonormal.

Introducing Lagrange multipliers $\lambda_1$ and $\lambda_2$, the
corresponding Euler--Lagrange equation is
\begin{equation}
    -(u^2\theta')'
    +2\lambda_1u^2\cos(2\theta)
    -2\lambda_2u^2\sin(2\theta)
    =0.
    \label{eq:analytical_theta_EL}
\end{equation}

The special form of the density allows the Euler--Lagrange equation
to be simplified further by using the cumulative density
\begin{equation}
    F(x)
    =
    \int_{-\infty}^{x}u^2(t)\,dt.
    \label{eq:analytical_cumulative_density}
\end{equation}
For~\eqref{eq:analytical_density}, direct integration gives
\begin{equation}
    F(x)
    =
    \frac{8}{\pi}\tan^{-1}\!\left(e^{z}\right),
    \qquad
    z=kx+c.
    \label{eq:analytical_F}
\end{equation}
In particular,
\begin{equation}
    u^2(x)
    =
    \frac{4k}{\pi}
    \sin\left(\frac{\pi F(x)}{4}\right).
    \label{eq:analytical_u_F}
\end{equation}

Since $F'(x)=u^2(x)>0$, the variable $F$ can be used as a new
independent variable, with
\begin{equation}
    \frac{d}{dx}
    =
    u^2\frac{d}{dF}.
\end{equation}
Hence,
\begin{equation}
    \theta'
    =
    u^2\frac{d\theta}{dF},
\end{equation}
and
\begin{equation}
    (u^2\theta')'
    =
    u^2\frac{d}{dF}
    \left(
        u^4\frac{d\theta}{dF}
    \right).
\end{equation}
Dividing~\eqref{eq:analytical_theta_EL} by $u^2$ therefore gives
\begin{equation}
    -\frac{d}{dF}
    \left(
        u^4\frac{d\theta}{dF}
    \right)
    +
    2\lambda_1\cos(2\theta)
    -
    2\lambda_2\sin(2\theta)
    =0.
    \label{eq:analytical_theta_EL_F_general}
\end{equation}
Using~\eqref{eq:analytical_u_F}, this becomes
\begin{equation}
    -\frac{d}{dF}
    \left[
        \frac{16k^2}{\pi^2}
        \sin^2\left(\frac{\pi F}{4}\right)
        \frac{d\theta}{dF}
    \right]
    +
    2\lambda_1\cos(2\theta)
    -
    2\lambda_2\sin(2\theta)
    =0.
    \label{eq:analytical_theta_EL_F}
\end{equation}

Equation~\eqref{eq:analytical_theta_EL_F} admits the simple solution
\begin{equation}
    \theta(F)
    =
    \frac{\pi}{4}F-\frac{\pi}{2}.
    \label{eq:analytical_theta_F}
\end{equation}
with multipliers
\begin{equation}
    \lambda_1=0,
    \qquad
    \lambda_2=\frac{k^2}{2}.
    \label{eq:analytical_multipliers}
\end{equation}
Thus, the angle field is obtained explicitly from the prescribed
density through its cumulative density. At this stage, however, this
solution is only known to be a stationary point of the constrained
minimization problem; its status as the global minimizer will be
established below.

Returning to the spatial coordinate,~\eqref{eq:analytical_F} and
~\eqref{eq:analytical_theta_F} give
\begin{equation}
    \theta(x)
    =
    2\tan^{-1}\!\left(e^{kx+c}\right)
    -\frac{\pi}{2}.
    \label{eq:analytical_theta}
\end{equation}
The corresponding trigonometric identities are
\begin{equation}
    \cos\theta(x)
    =
    \operatorname{sech}(kx+c),
    \qquad
    \sin\theta(x)
    =
    \tanh(kx+c),
    \label{eq:analytical_angle_identities}
\end{equation}
and
\begin{equation}
    \theta'(x)
    =
    k\operatorname{sech}(kx+c).
\end{equation}

The resulting angle field is shown in
Figure~\ref{fig:analytical_theta}, illustrating how the prescribed
density is mapped to an explicit angle field through its cumulative
density.

\begin{figure}
    \centering
    \includegraphics[width=0.82\linewidth]{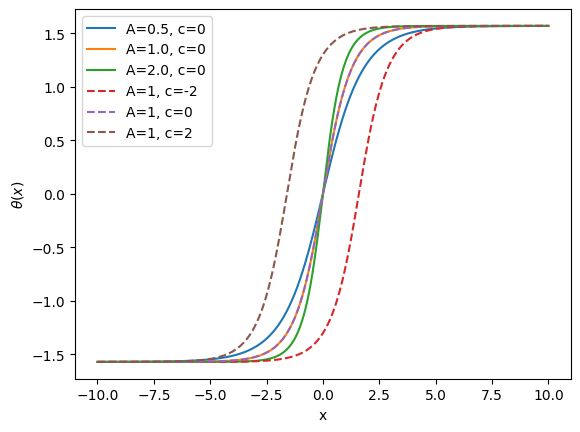}
    \caption{The angle field $\theta(x)$ generated from the cumulative
    density for different members of the density family
    \eqref{eq:analytical_density}.}
    \label{fig:analytical_theta}
\end{figure}

\subsection{From the angle field to the orbitals}

The corresponding two orbitals are given by the trigonometric
parametrization
\begin{equation}
    \phi_1(x)
    =
    \frac{u_{A,c}(x)}{\sqrt{2}}\sin\theta(x),
    \qquad
    \phi_2(x)
    =
    \frac{u_{A,c}(x)}{\sqrt{2}}\cos\theta(x).
    \label{eq:analytical_general_orbitals}
\end{equation}

Using~\eqref{eq:analytical_u} and the identities in
\eqref{eq:analytical_angle_identities}, the orbitals can be written
explicitly as
\begin{equation}
\begin{split}
    \phi_1(x)
    &=
    \sqrt{\frac{2k}{\pi}}\,
    \operatorname{sech}^{1/2}z\,\tanh z,\\
    \phi_2(x)
    &=
    \sqrt{\frac{2k}{\pi}}\,
    \operatorname{sech}^{3/2}z.
\end{split}
\label{eq:analytical_orbitals}
\end{equation}
The orbital $\phi_2$ is strictly positive, whereas $\phi_1$ has one
node at $z=0$. It is straightforward to verify that
\begin{equation}
    2\left(\phi_1^2+\phi_2^2\right)
    =
    u_{A,c}^2
    =
    \rho_{A,c},
    \label{eq:analytical_density_reconstruction}
\end{equation}
and
\begin{equation}
    \int_{\mathbb{R}}\phi_1^2\,dx=1,
    \qquad
    \int_{\mathbb{R}}\phi_2^2\,dx=1,
    \qquad
    \int_{\mathbb{R}}\phi_1\phi_2\,dx=0.
    \label{eq:analytical_orthonormality}
\end{equation}
Hence, the two functions form an orthonormal pair whose doubly
occupied density is exactly the prescribed density. An example of the
resulting orbitals together with the corresponding density is shown in Figure~\ref{fig:analytical_orbitals}.

\begin{figure}
    \centering
    \includegraphics[width=0.82\linewidth]{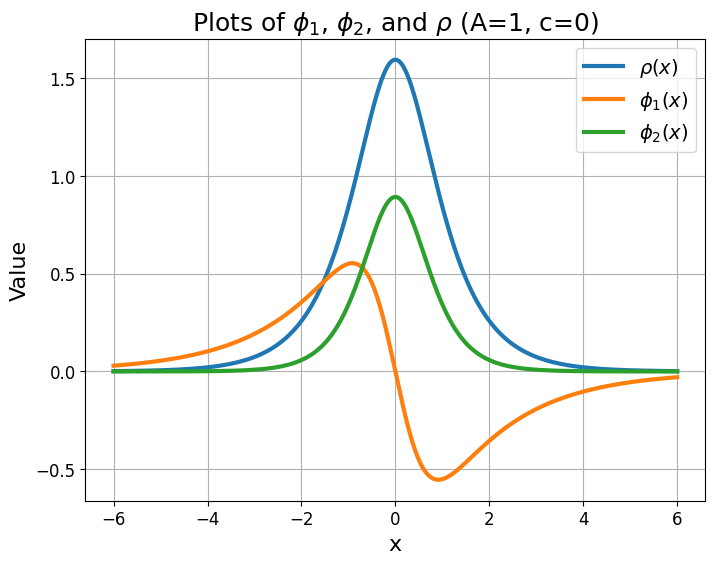}
    \caption{The prescribed density $\rho_{A,c}$ and the corresponding
    orbitals $\phi_1$ and $\phi_2$ obtained from the angle-field
    construction.}
    \label{fig:analytical_orbitals}
\end{figure}

\subsection{Determining the common potential}

The angle-field construction above provides a stationary solution of the constrained minimization problem. To establish that this solution corresponds to the Kohn--Sham ground state, we now determine whether the density-generated orbitals are the two lowest eigenfunctions of a common local potential. For the restricted Kohn--Sham minimization, the Euler--Lagrange
equations take the form
\begin{equation}
    \left(
    -\frac{1}{2}\frac{d^2}{dx^2}+v(x)
    \right)\phi_i
    =
    \varepsilon_i\phi_i,
    \qquad i=1,2.
    \label{eq:analytical_EL}
\end{equation}

Substitution of the orbitals in~\eqref{eq:analytical_orbitals} into
~\eqref{eq:analytical_EL} shows that both orbitals are eigenfunctions
of the same local potential. Choosing the additive constant such that
$v(x)\to0$ as $|x|\to\infty$, we obtain
\begin{equation}
    v(x)
    =
    -\frac{15}{8}k^2\operatorname{sech}^2(kx+c)
    =
    -\frac{15A\pi}{16}
    \operatorname{sech}^2\!\left(
        \sqrt{\frac{A\pi}{2}}x+c
    \right).
    \label{eq:analytical_potential}
\end{equation}
The corresponding eigenvalues are
\begin{equation}
    \varepsilon_0=-\frac{9}{8}k^2,
    \qquad
    \varepsilon_1=-\frac{1}{8}k^2,
    \label{eq:analytical_eigenpairs}
\end{equation}
and hence
\begin{equation}
    \varepsilon_1-\varepsilon_0
    =
    k^2
    =
    \frac{A\pi}{2}.
    \label{eq:analytical_eigenvalue_gap}
\end{equation}

The resulting common local potential is shown in
Figure~\ref{fig:analytical_potential}. It has the form of a modified
Pöschl--Teller potential
\begin{equation}
    v(x)
    =
    -\frac{k^2}{2}\lambda(\lambda+1)
    \operatorname{sech}^2(kx+c),
    \qquad
    \lambda=\frac{3}{2}.
\end{equation}

\begin{figure}
    \centering
    \includegraphics[width=0.82\linewidth]{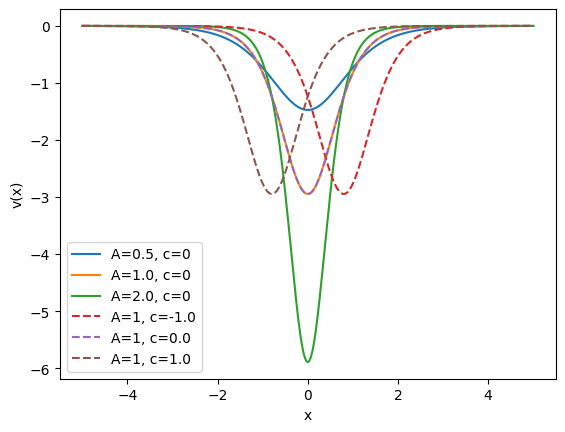}
    \caption{The common local potential associated with the
    density-generated orbitals.}
    \label{fig:analytical_potential}
\end{figure}

The orbital $\phi_2$ is strictly positive and therefore has no nodes,
whereas $\phi_1$ has exactly one node at $z=0$. By the one-dimensional
nodal theorem, $\phi_2$ is therefore the ground-state orbital and
$\phi_1$ is the first excited-state orbital of the potential
~\eqref{eq:analytical_potential}. Thus, the orbitals obtained from the
angle-field construction are precisely the two lowest orbitals of a
common local potential. This establishes that the stationary
angle-field solution is the global minimizer of the constrained
kinetic-energy problem. In particular,
\begin{equation}
    \rho_{A,c}
    =
    2\left(\phi_1^2+\phi_2^2\right)
\end{equation}
is the ground-state density of the corresponding four-electron
non-interacting system.

\subsection{Non-interacting kinetic energy}

Since the orbitals obtained above are the two lowest orbitals of the
common local potential~\eqref{eq:analytical_potential}, the prescribed
density is non-interacting $v$-representable and the angle-field
minimum equals the restricted Kohn--Sham kinetic energy,
\begin{equation}
    T^+(u_{A,c})
    =
    T_{\mathrm{RKS}}(\rho_{A,c}).
\end{equation}

The kinetic energy can be evaluated directly from the orbitals:
\begin{equation}
    T_{\mathrm{RKS}}(\rho_{A,c})
    =
    2\sum_{i=1}^{2}
    \frac12
    \int_{\mathbb R}
    |\phi_i'(x)|^2\,dx.
    \label{eq:analytical_TKS_orbitals}
\end{equation}
Using~\eqref{eq:analytical_orbitals}, we obtain
\begin{equation}
    \frac12\int_{\mathbb R}|\phi_1'|^2\,dx
    =
    \frac{11}{32}k^2,
    \qquad
    \frac12\int_{\mathbb R}|\phi_2'|^2\,dx
    =
    \frac{9}{32}k^2.
\end{equation}
Therefore,
\begin{equation}
    T_{\mathrm{RKS}}(\rho_{A,c})
    =
    \frac54 k^2
    =
    \frac{5A\pi}{8}.
    \label{eq:analytical_TKS}
\end{equation}

It is also easy to see that the same result follows directly from the angle-field decomposition.
Since
\begin{equation}
    \theta'(x)
    =
    \frac{\pi}{4}u_{A,c}^2(x),
\end{equation}
the kinetic energy can be written as
\begin{equation}
    T^+(u_{A,c})
    =
    \frac12
    \int_{\mathbb R}(u_{A,c}')^2\,dx
    +
    \frac{\pi^2}{32}
    \int_{\mathbb R}u_{A,c}^6\,dx.
    \label{eq:analytical_energy_functional}
\end{equation}
For~\eqref{eq:analytical_u},
\begin{equation}
    \int_{\mathbb R}(u_{A,c}')^2\,dx
    =
    \frac{k^2}{2},
    \qquad
    \int_{\mathbb R}u_{A,c}^6\,dx
    =
    \frac{32k^2}{\pi^2}.
\end{equation}
Consequently,
\begin{equation}
    T^+(u_{A,c})
    =
    \frac14k^2+k^2
    =
    \frac54k^2,
\end{equation}
in agreement with the direct orbital calculation in
~\eqref{eq:analytical_TKS}.

In summary, for the density family $\rho_{A,c}$, the angle-field formulation provides an explicit density-to-potential construction. The prescribed density determines the cumulative-density angle field, which in turn yields the orthonormal orbitals and the associated local potential. The nodal structure of these orbitals establishes that they are the two lowest eigenfunctions of this potential, thereby confirming that the stationary angle field is the global minimizer and that $\rho_{A,c}$ is non-interacting $v$-representable. This example illustrates, in a fully analytical setting, how the proposed angle-field formulation directly connects a prescribed density to its corresponding Kohn--Sham kinetic energy and local potential.

\section{Numerical Examples in One Dimension}
\label{sec:numerical_examples}
We present in this section a numerical demonstration of the ideas developed in this work. We restrict ourselves to the one-dimensional context and reserve a numerical study of the three-dimensional case for a future work. The primary goal of this numerical section is to demonstrate how the extension to $T_\text{RKS}(\rho)$ developed in this work can be computed for \emph{any} $N$-representable $\rho$ numerically. We wish to clarify that we do not attempt to provide a closed form expression for $T_\text{RKS}(\rho)$. We envisage this procedure to result in a sufficiently large dataset from which a suitable expression for $T_\text{RKS}(\rho)$ can be learned---we reserve this for a future work. Towards the end of this section, we consider densities which are not $N$-representable and show 1D slices of the extended non-interacting kinetic energy functional $T^+(\rho)$. 

\subsection{1D optimization problem for non-interacting kinetic energy}
Let $\rho \in H^1\cap L^\infty(\mathbb{R})$ be a non-negative function compactly supported in $[0,1]$ and let 
\begin{equation*}
N = \text{ceil}\left( \int_0^1 \frac{\rho}{2}\right).
\end{equation*}
We introduce the angle fields $\Theta = (\theta_1, \theta_2, ..., \theta_{N-1}) \in H^1((0,1))$ and solve the minimization problem~\eqref{eq:Tplusinfnew}  involving $\Theta$ by minimizing the following functional: 
\begin{equation*}
    F(\Theta) = \frac{1}{2}\int_0^1 f(\Theta),\\ 
\end{equation*}
subject to the constraint
$$ C_{ij}(\Theta) = \bigg\lvert\int_0^1 \varphi^{\Theta}_i \varphi^{\Theta}_j -\delta_{ij}\bigg\rvert < \delta , 1 \le i, j \le N, \text{ except for }i=j=N,$$ for some $\delta > 0$ and where
$$
f(\Theta) =  \rho  (\theta_1^{'})^2  
    +  \rho \Big(  \cos^2{\theta_1} (\theta_2^{'})^2  +  \cos^2{\theta_1}\cos^2{\theta_2}(\theta_3^{'})^2 + \ldots + \cos^2{\theta_1}\cos{\theta_2}^2 \ldots \cos^2{\theta_{N-2}}  (\theta_{N-1}^{'})^2 \Big).
$$

Note that $T^+(\rho) = \lim_{\delta \downarrow 0^+}\frac{1}{2} \int_0^1 | (\sqrt{\rho})'|^2 + F(\Theta^{0, \delta})$ where $\Theta^{0, \delta}$ is the minimizer of $F$ for a given $\delta$. 

\begin{remark}
The minimization problem \eqref{eq:Tplusinfnew} requires a two level minimization over $\delta$. However, since $F(\Theta^{0, \delta})$ monotonically increases with as $\delta \downarrow 0^+$,  for computational expediency, we fix $\delta$ to be machine precision as that is the best one could do numerically. Hence the constraints become $C_{ij} \approx 0$, $\Theta^{0}$ is the minimizer of $F$ with these constraints and $T^+(\rho) \approx \frac{1}{2} \int_0^1 | (\sqrt{\rho})'|^2 + F(\Theta)$. We later see that this approximation is very close to $KSDFT$ results.  
\end{remark}

To solve this minimization problem, we approximate the angle fields using Radial Basis Functions (RBFs):
$$
\theta_i(x) = \sum_{j=1}^M t_{ij}\sigma(|x - \xi_j|), \quad 1 \le i < N,
$$
where $\sigma:\mathbb{R}_+ \to \mathbb{R}$ is the parent RBF, $\{\xi_j\}_{j=1}^M$ are the centers of $M$ nodes of the RBF approximation, and $t_{ij}$ are the corresponding coefficients. We refer back to the final remark in Section~\ref{sec:minprob} for validating the choice of the RBF approximation. We will use the shorthand notation $\sigma_j$ in place of $\sigma(|\cdot - \xi_j|)$. For the results presented in this section, we choose the multiquadric RBF function $\sigma = \sqrt{1 + (x/\epsilon)^2}$, with $\epsilon=0.1$. With this choice of approximation for $\Theta$, both $F$ and $C$ are functions only of $t_{ij}$'s, and can be solved using any appropriate nonlinear programming algorithm. For the numerical result presented in this section, we used the Sequential Least Squares Quadratic Programming (SLSQP) method \cite{Kraft88} as implemented in the \emph{SciPy} library \footnote{\url{https://docs.scipy.org/doc/scipy/reference/generated/scipy.optimize.minimize.html}}. A listing of the various derivatives needed for a gradient-based minimization algorithm are provided in Appendix~\ref{app:FCgrad}. 

For the local minimization problem using SLSQP, the choice of initial conditions plays a crucial role in determining the outcome of the optimization. Since SLSQP relies on local minimization, different initial conditions can lead to convergence towards different local minima, affecting the accuracy and consistency of the results. To account for this, we perform simulations using 100 randomly chosen initial conditions---the results are shown in  
Figure~\ref{fig:rand_ic} as a log histogram of the computed energies for four randomly generated $N$-representable densities $\rho$, with $N=3$. The dependence on the choice of initial conditions makes the theory algorithm-dependent in practice. We present in the next section a reliable choice of initial conditions to circumvent this issue.

\begin{figure}[htbp]
    \centering
    \begin{subfigure}[b]{0.45\textwidth}
        \includegraphics[width=\textwidth]{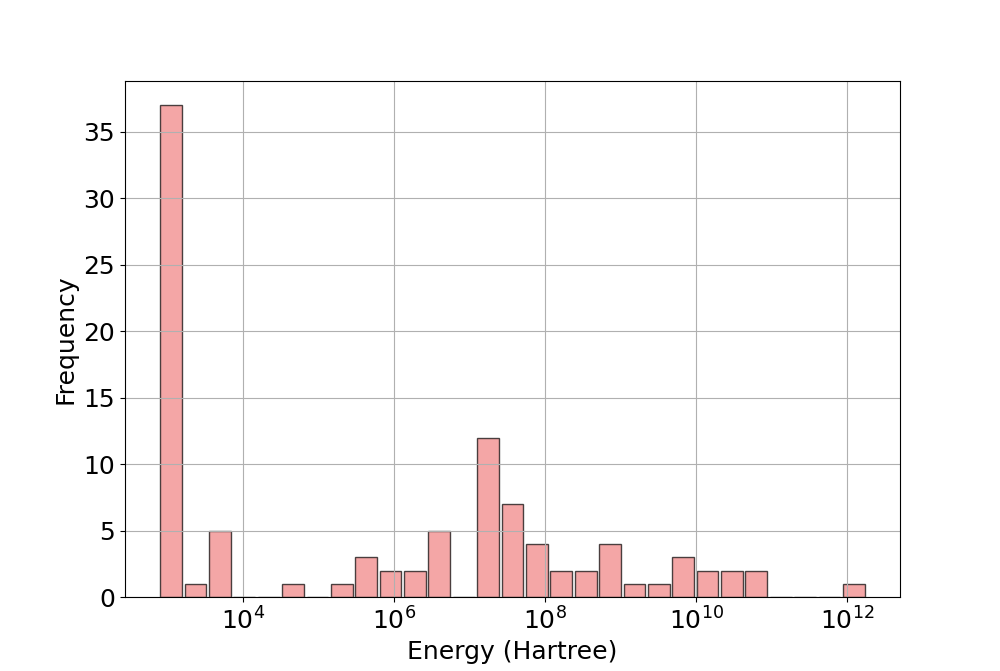}
    \end{subfigure}
    \begin{subfigure}[b]{0.45\textwidth}
        \includegraphics[width=\textwidth]{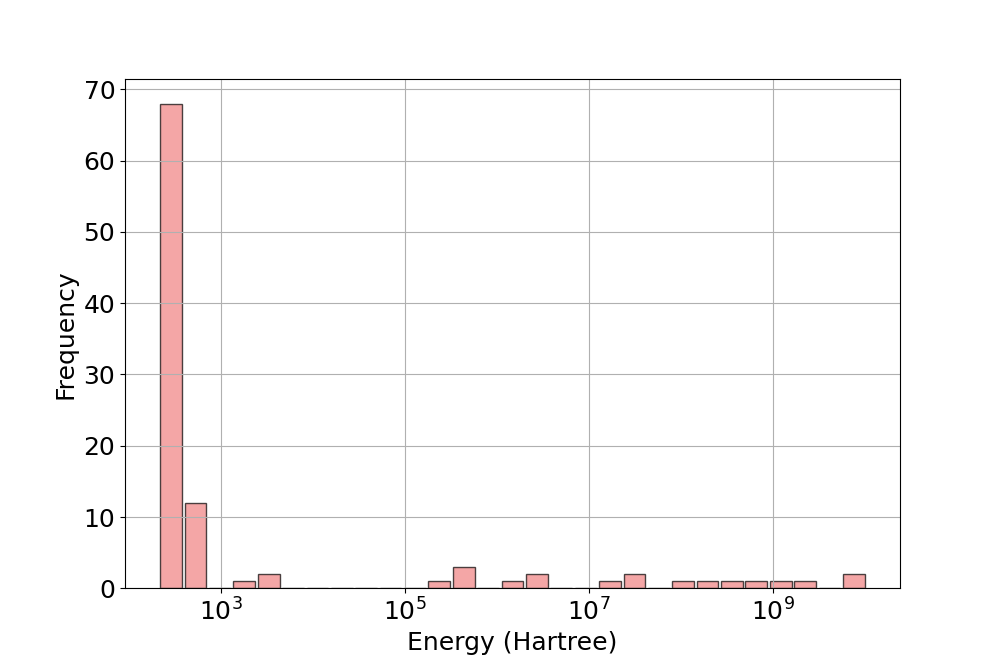}
    \end{subfigure}
    \begin{subfigure}[b]{0.45\textwidth}
        \includegraphics[width=\textwidth]{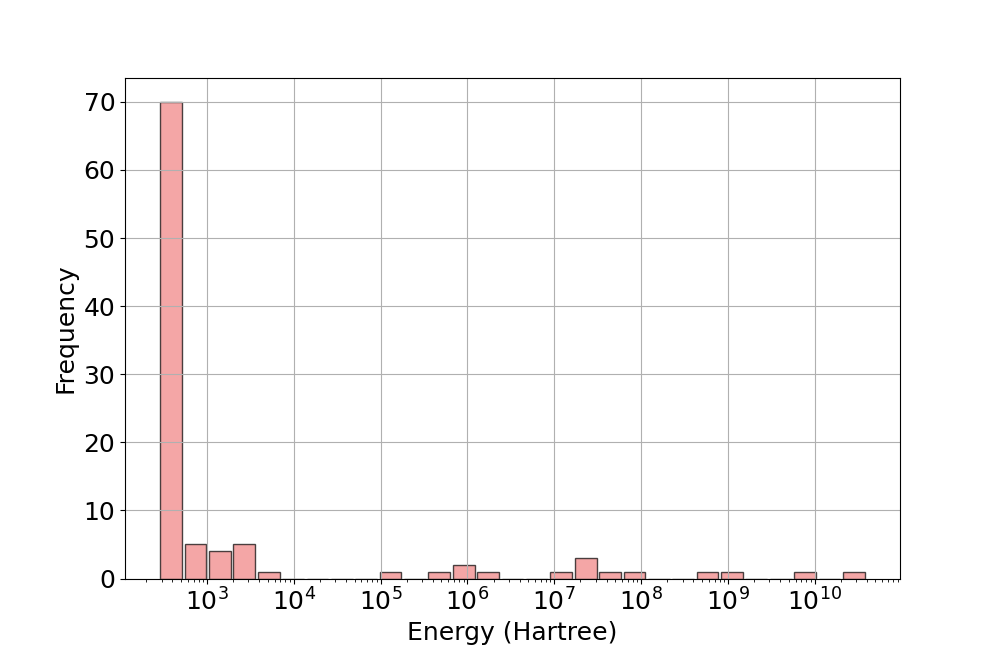}
    \end{subfigure}
    \begin{subfigure}[b]{0.45\textwidth}
        \includegraphics[width=\textwidth]{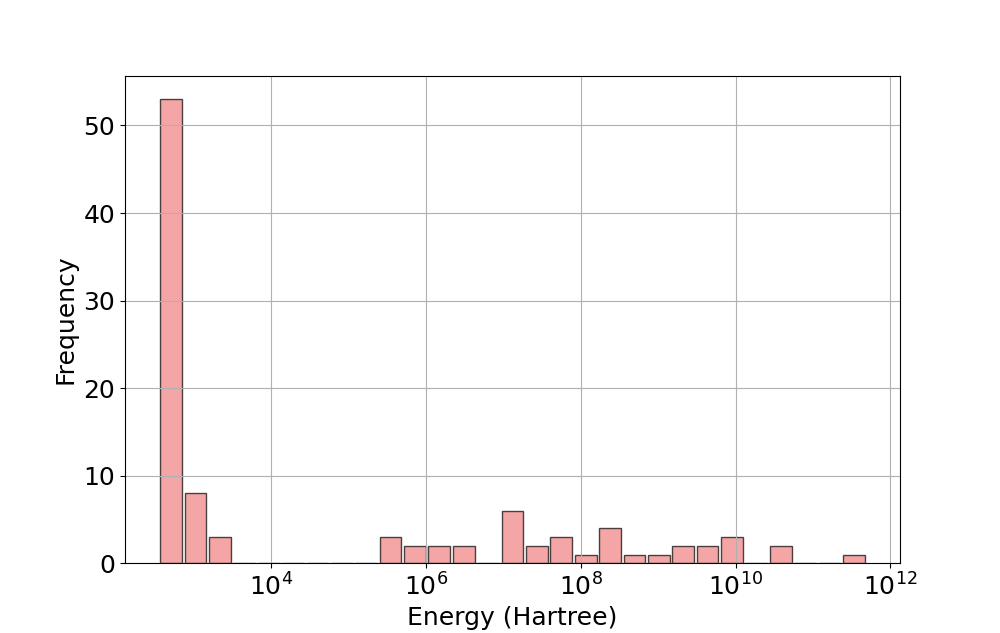}
    \end{subfigure}
    \caption{Histograms of non-interacting kinetic energy values, obtained from various initial conditions, for a given electron density, $\rho$, in each plot. Here, $ \int_{\mathbb{R}} \rho = 2N = 6$.}
    \label{fig:rand_ic}
\end{figure}
\subsection{Numerical Validation with KSDFT models}
To validate our approach, we test it using a one-dimensional Kohn-Sham system confined to the interval $[0,1]$. For a randomly generated 1D potential \( v:[0,1] \to \mathbb{R} \), we compute the first \( N \) eigenvalues and associated eigenvectors of the operator $-\frac{1}{2}\Delta + v $ on the domain \( (0,1) \) with Dirichlet boundary conditions. The non-interacting kinetic energy, $T$, of the system is then computed using equation~\eqref{eq:T_rks}. The electron density \( \rho \) is evaluated as twice the sum of the squares of the first $N$ normalized eigenvectors. Finally, we evaluate the extended non-interacting kinetic energy \( T^+ \) using this density by solving the optimization problem presened in the previous section. We validate our method by checking that $T^+$ matches with the reference value $T$ with high accuracy.

\subsubsection{Special choice of initial condition}
For the Kohn-Sham ground-state electron densities, we found that the initial conditions corresponding to the orbitals of a particle in a box produce results with desired accuracy among all the tested initial conditions. To illustrate this, we fix \( N = 3 \), discretize the domain into $1000$ points, and set the number of radial basis functions (RBF) to $20$. Figure~\ref{fig:ic_log_histogram_rho_gs} presents a log histogram showing the distribution of kinetic energies obtained from random initial conditions for four different densities $\rho$, generated from four randomly generated potentials. The special solution corresponding to an particles-in-a-box initial conditions are highlighted to illustrate their effectiveness in producing the best numerical solution. 
\begin{figure}[htbp]
    \centering
    \begin{subfigure}[b]{0.45\textwidth}
        \includegraphics[width=\textwidth]{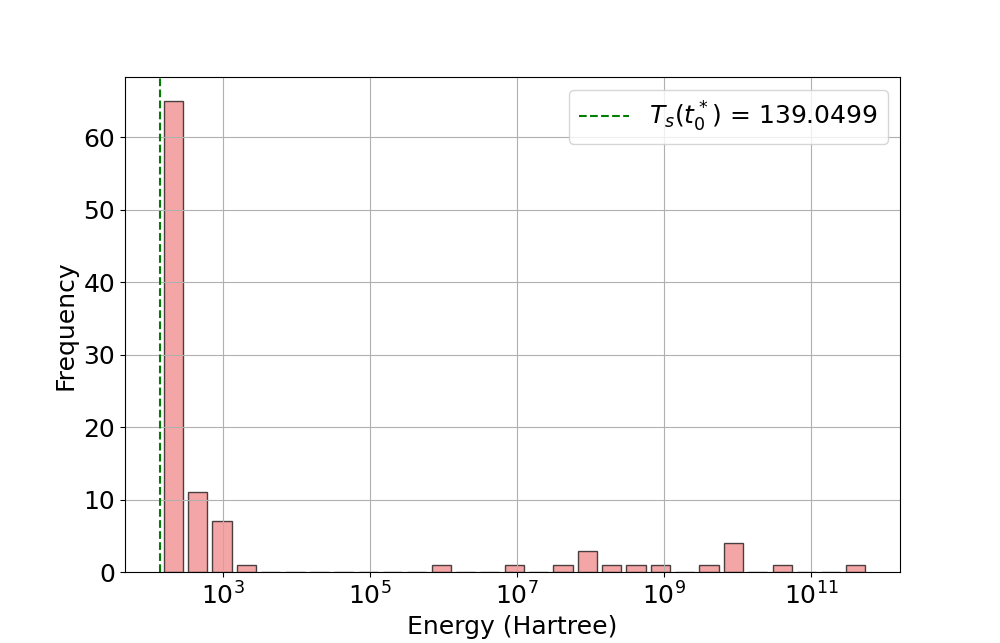}
    \end{subfigure}
    \begin{subfigure}[b]{0.45\textwidth}
        \includegraphics[width=\textwidth]{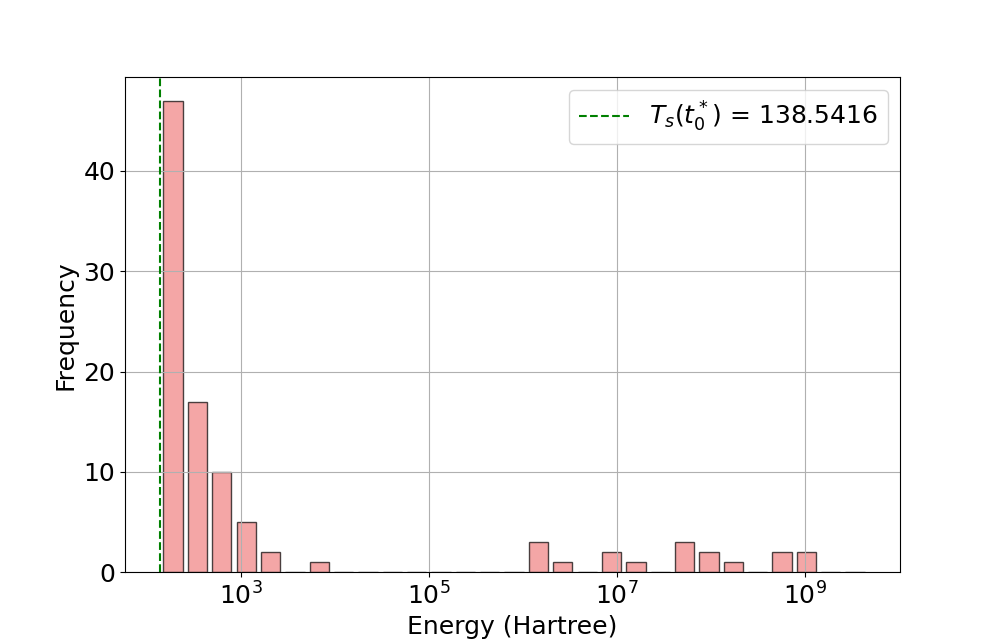}
    \end{subfigure}
    \begin{subfigure}[b]{0.45\textwidth}
        \includegraphics[width=\textwidth]{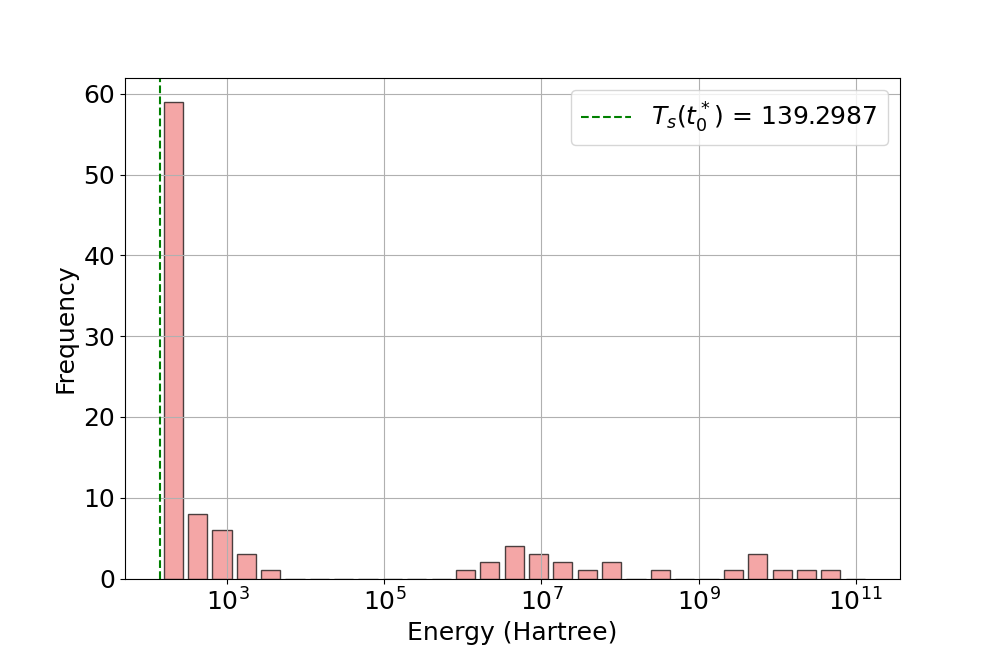}
    \end{subfigure}
    \begin{subfigure}[b]{0.45\textwidth}
        \includegraphics[width=\textwidth]{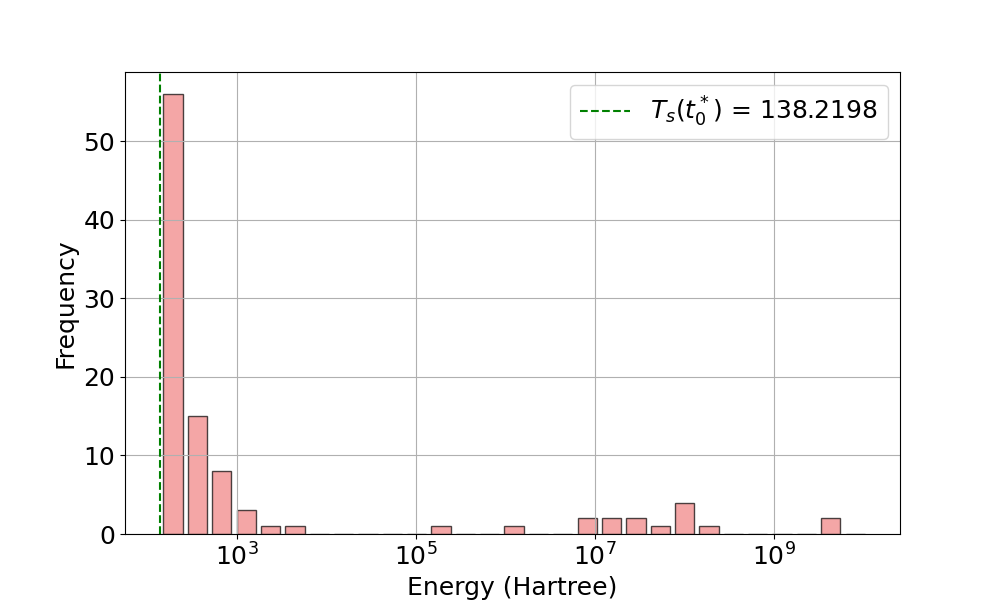}
    \end{subfigure}
    \caption{Histograms of non-interacting kinetic energy values, obtained from various initial conditions, for a given electron density, $\rho$, in each plot. Here, $ \int_{\mathbb{R}} \rho = 2N = 6$ and $t_0^*$ is initial condition obtained from the orbitals of particle in a box. The value $T_s(t_0^*)$, , representing the non-interacting kinetic energy evaluated using $t_0^*$, indicates that this particular choice of initial condition yields a kinetic energy close to the minimum.}
\label{fig:ic_log_histogram_rho_gs}
\end{figure}
Figure~\ref{fig:ic_histogram_rho_gs} shows a zoomed-in view of the histograms in Figure~\ref{fig:ic_log_histogram_rho_gs}, and, in particular, plots the histogram of the relative kinetic energy computed using our algorithm with respect to a Kohn-Sham calculation. It is evident that the particle-in-a-box initialization produces kinetic energies that are lower than those obtained from all other random initializations. We thus fix this special choice of initial conditions for all the simulations presented in the rest of this section. Note that this section focuses on demonstrating the effectiveness of the special initial condition rather than comparing with KS-DFT. Therefore, we have not used the optimal radial basis function or domain discretization.
\begin{figure}[htbp]
    \centering
    \begin{subfigure}[b]{0.45\textwidth}
        \includegraphics[width=\textwidth]{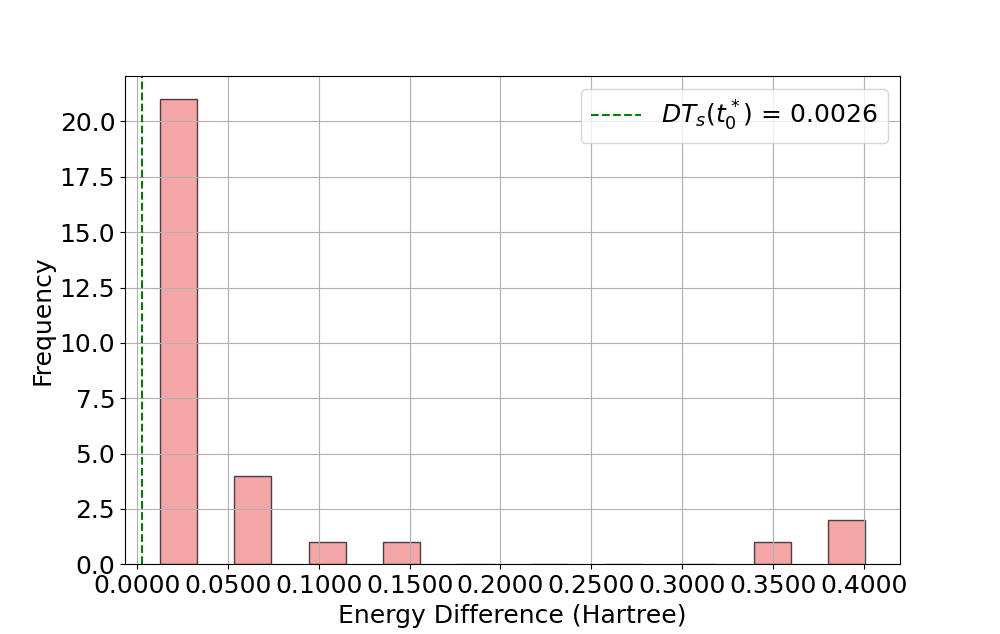}
    \end{subfigure}
    \begin{subfigure}[b]{0.45\textwidth}
        \includegraphics[width=\textwidth]{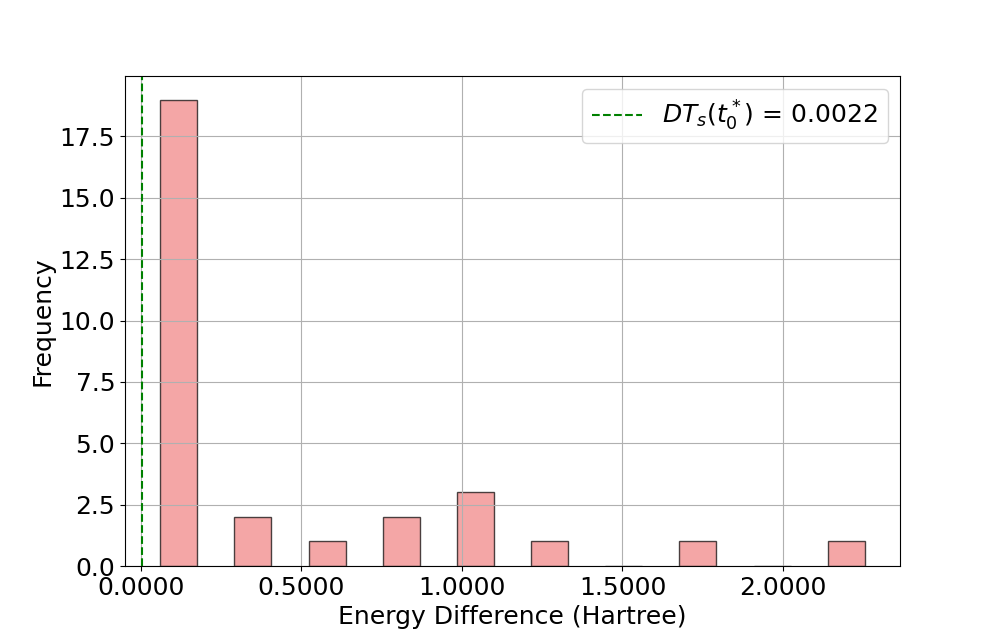}
    \end{subfigure}
    \begin{subfigure}[b]{0.45\textwidth}
        \includegraphics[width=\textwidth]{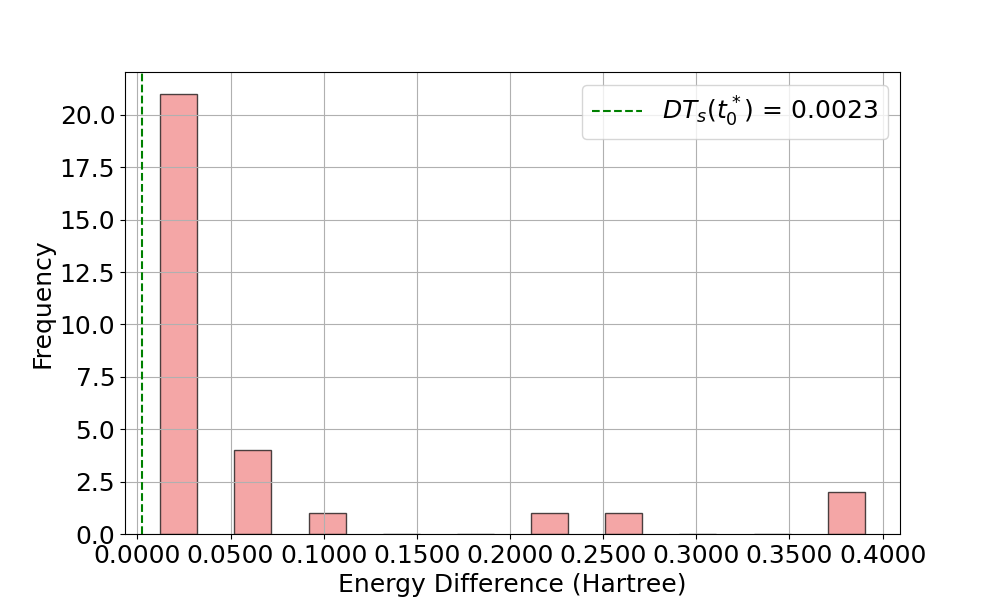}
    \end{subfigure}
    \begin{subfigure}[b]{0.45\textwidth}
        \includegraphics[width=\textwidth]{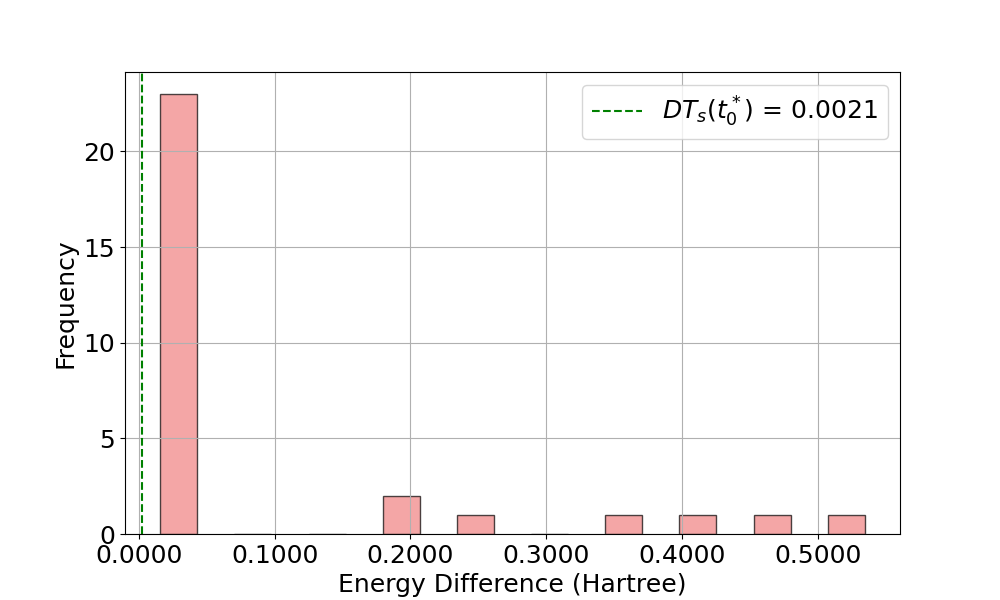}
    \end{subfigure}
    \caption{Histograms showing the lowest ten percent of non-interacting kinetic energy values, obtained from various initial conditions, relative to the non-interacting kinetic energy of the Kohn Sham system for a given electron density, $\rho$, in each plot. Here, $ \int_{\mathbb{R}} \rho = 2N = 6$, $t_0^*$ is initial condition obtained from the orbitals of particle in a box. The value $DT_s(t_0^*)$ is difference between non-interacting kinetic energy from the special initial condition and the non-interacting kinetic energy of the Kohn Sham system.}
\label{fig:ic_histogram_rho_gs}
\end{figure}

\subsubsection{Comparison with Kohn-Sham models}
We now perform a detailed numerical comparison of the kinetic energies computed using the methods developed in this work and the standard Kohn-Sham model. We generate 500, 250, and 100 Kohn-Sham density samples for \( N = 2, 3, 4 \), respectively. For each case, different resolutions of the domain discretization (\( n_e \)) and numbers of radial basis functions (\( n_{rbf} \)) are employed to achieve desired accuracy. Figures~\ref{fig:N2}, \ref{fig:N3}, and \ref{fig:N4} demonstrate that the error between the computed kinetic energy and the Kohn sham solution, per electron, remains within $1$ kcal/mol across these samples.

\subsubsection{Computational Cost}
For each value of \( N \), we also compute the average runtime and average total number of iterations required for SLSQP convergence. These results are summarized in Table~\ref{tab:avg_scaling}, which indicates that the algorithm exhibits \( \mathcal{O}(N^5) \) scaling. We emphasize that the goal of this work is not to develop a fast algorithm, but rather to demonstrate that the proposed approach for evaluating non-interacting kinetic energy is valid and produces results consistent with standard KSDFT calculations.
\begin{figure}[ht]
    \centering
    \includegraphics[width=0.6\textwidth]{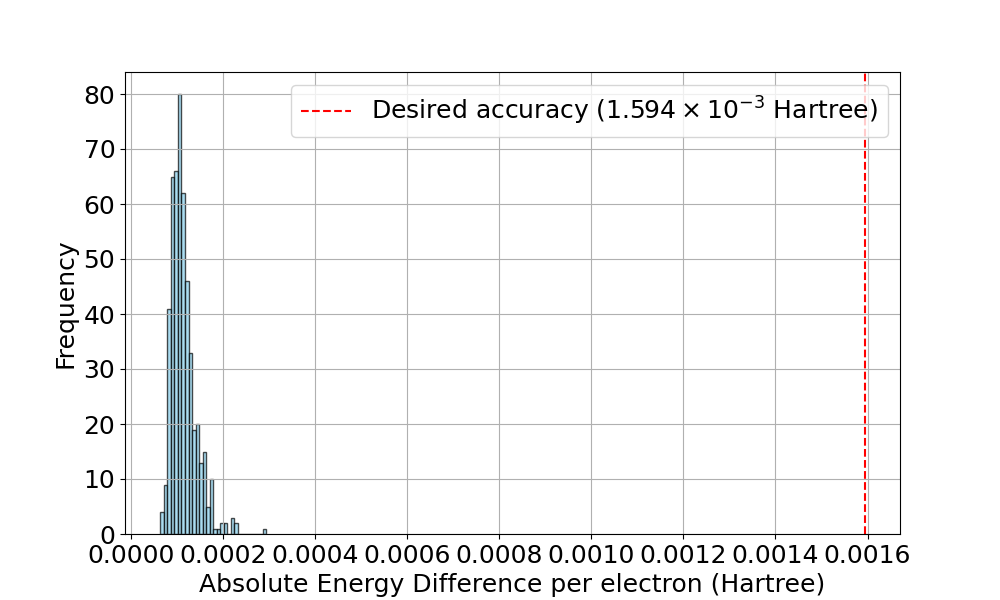}
    \caption{Error between the Kohn Sham non-interacting kinetic energy and non-interacting kinetic energy obtained from our minimisation method with special initial condition, per electron, for a given electron density, $\rho$, such that $ \int_{\mathbb{R}} \rho = 2N = 4$. Here, $n_e =1000$ and $n_{rbf} =15$ is used.}
    \label{fig:N2}
\end{figure}

\begin{figure}[ht]
    \centering
    \includegraphics[width=0.6\textwidth]{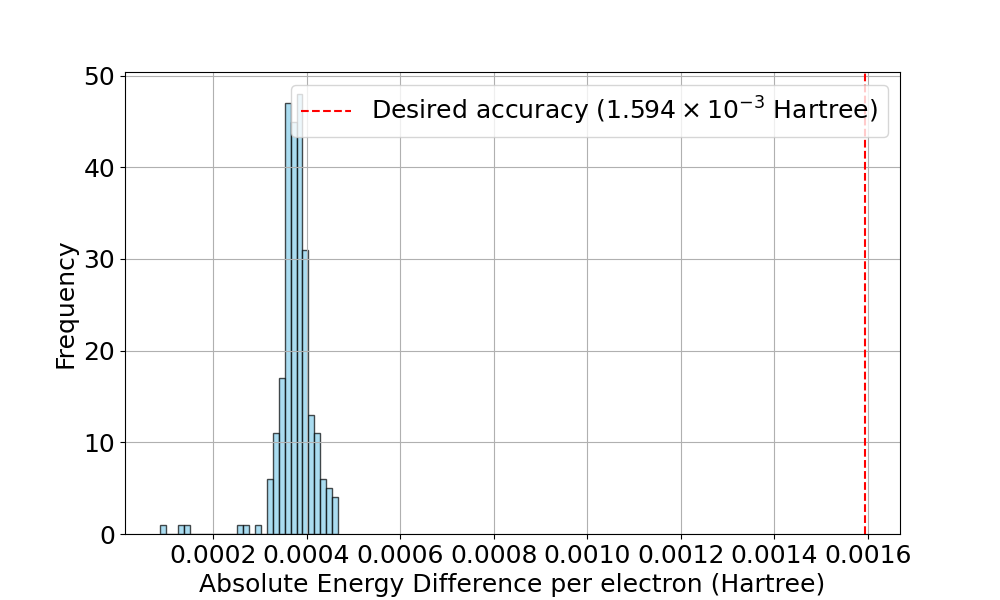}
    \caption{Error between the Kohn Sham non-interacting kinetic energy and non-interacting kinetic energy obtained from our minimisation method with special initial condition, per electron, for a given electron density, $\rho$, such that $ \int_{\mathbb{R}} \rho = 2N = 6$. Here, $n_e =1000$ and $n_{rbf} =20$ is used.}
    \label{fig:N3}
\end{figure}

\begin{figure}[ht]
    \centering
    \includegraphics[width=0.6\textwidth]{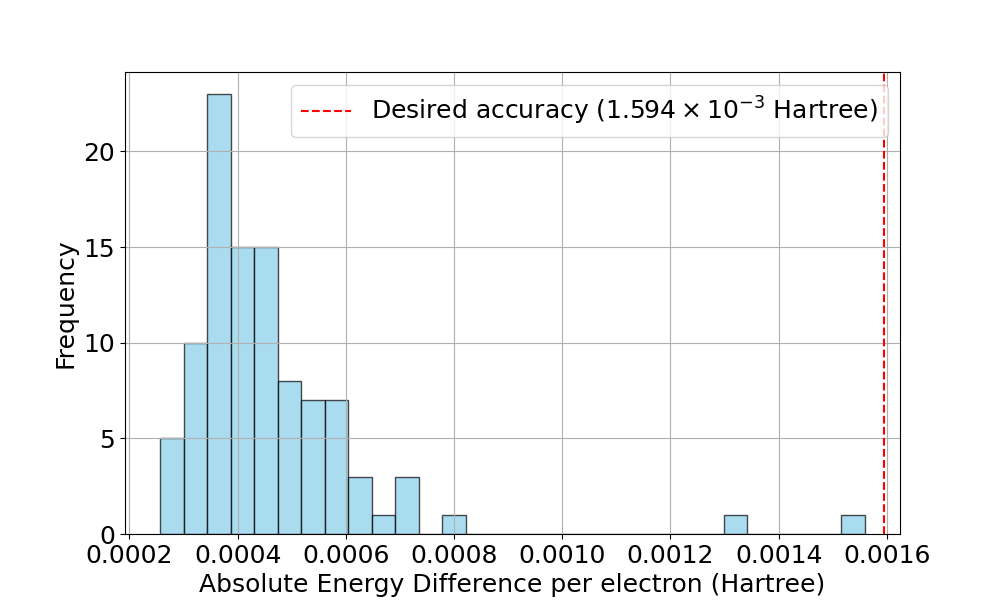}
    \caption{Error between the Kohn Sham non-interacting kinetic energy and non-interacting kinetic energy obtained from our minimisation method with special initial condition, per electron, for a given electron density, $\rho$, such that $ \int_{\mathbb{R}} \rho = 2N = 8$. Here, $n_e =2000$ and $n_{rbf} =70$ is used.}
    \label{fig:N4}
\end{figure}

\begin{table}[ht]
    \centering
    \begin{tabular}{|c|c|c|c|}
        \hline
        $N$ & Avg. Iterations & Avg. Time (s) & Relative time \\
        \hline
        2 & 22 & 4.6 & 1 \\
        3 & 100 & 41 & 8.9 \\
        4 & 100 & 148 & 32.2 \\
        \hline
    \end{tabular}
    \caption{Average SLSQP iterations and time for different values of $N$. The relative time refers to the time measured with respect to the average simulation time for $N=2$.}
    \label{tab:avg_scaling}
\end{table}

\subsubsection{Rate of Convergence}
We generate 10 Kohn-Sham density samples for \( N = 2, 3, 4 \), maintaining the same domain discretization resolution (\( n_e \)) and number of radial basis functions (\( n_{rbf} \)) as previously used to achieve desired accuracy in each case. This time, our goal is to examine the convergence rate of the SLSQP algorithm with respect to the iteration number when minimizing $F(\Theta)/2$ for each case. Figures~\ref{fig:Nc2}, \ref{fig:Nc3}, and \ref{fig:Nc4} illustrate the convergence behavior. We were not able to obtain a clear convergence order based on the results, but they are presented here for the sake of completeness.
\begin{figure}[ht]
    \centering
 \includegraphics[width=0.6\textwidth]{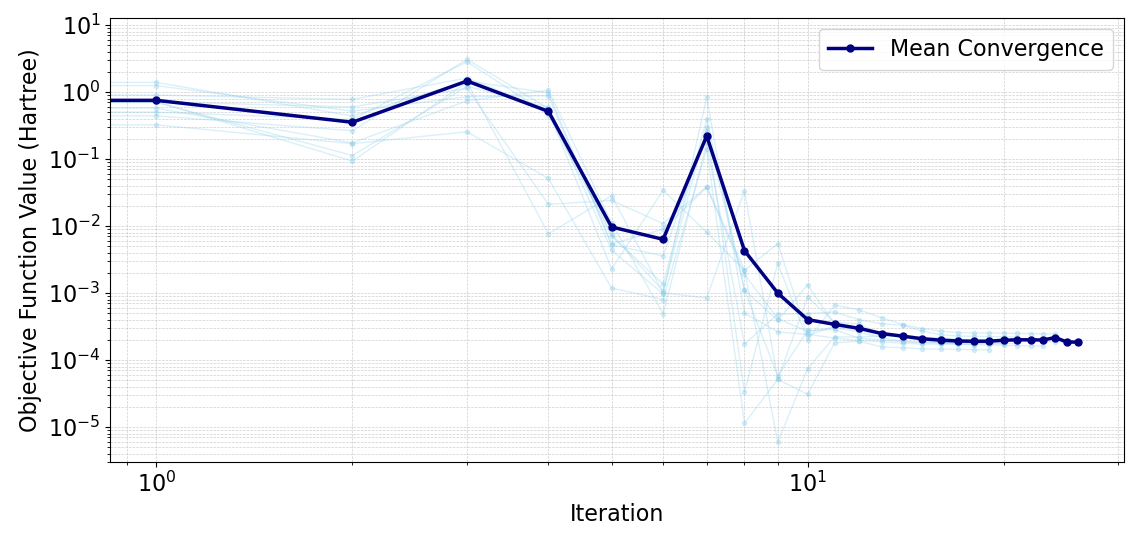}
    \caption{Rate of convergence of SLSQP for $N=2$. Here $n_e =1000$ and $n_{rbf} =15$ is used.}
    \label{fig:Nc2}
\end{figure}

\begin{figure}[ht]
    \centering
    \includegraphics[width=0.6\textwidth]{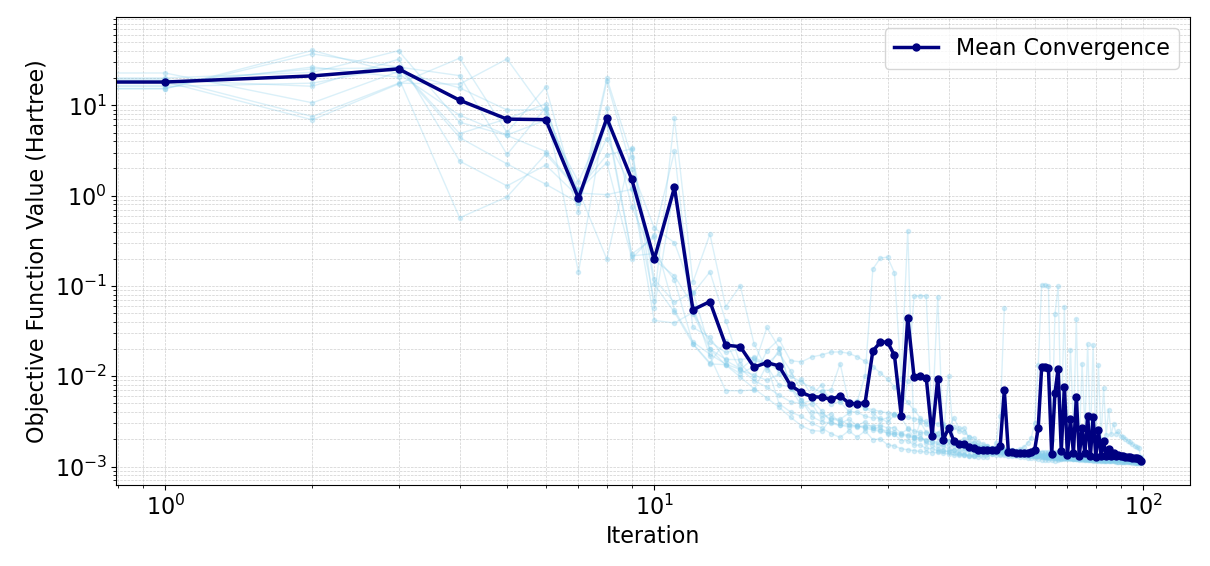}
    \caption{Rate of convergence of SLSQP for $N=3$. Here $n_e =1000$ and $n_{rbf} =20$ is used.}
    \label{fig:Nc3}
\end{figure}

\begin{figure}[ht]
    \centering
    \includegraphics[width=0.6\textwidth]{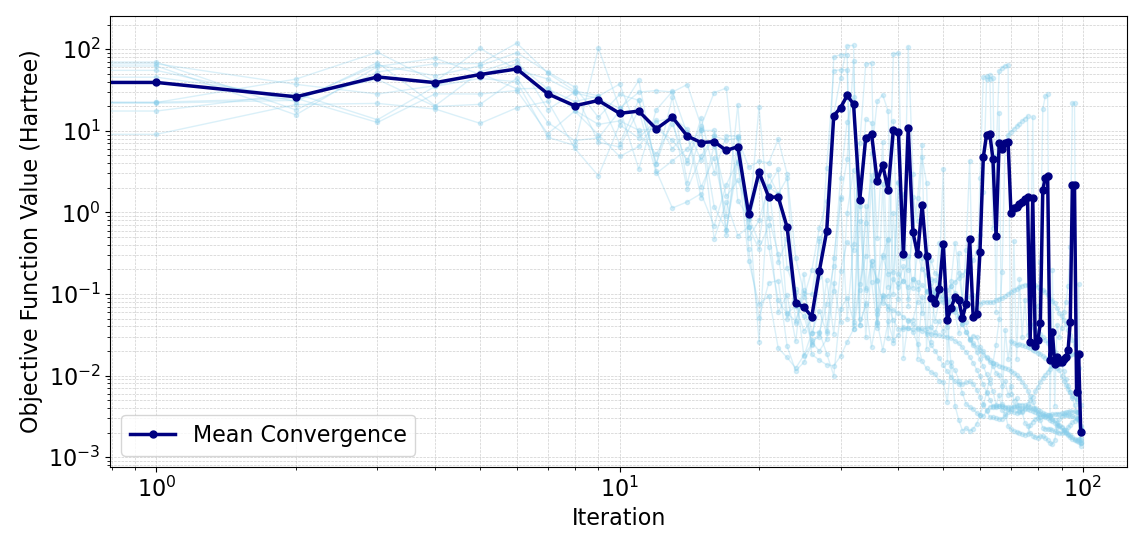}
    \caption{Rate of convergence of SLSQP for $N=4$. Here $n_e =2000$ and $n_{rbf} =70$ is used.}
    \label{fig:Nc4}
\end{figure}

\subsection{A numerical studies of the properties of $T^+$}
The extended non-interacting kinetic energy $T^+$ was introduced in this work primarily for the sake of mathematical analysis. We present in this section some numerical simulations to investigate the convexity of \( T^+ \). For this section, we use $\rho$ instead of $u$ as the primary variable of $T^+$, but retain the functional symbol for notational convenience.  Let \( \rho_1 \) and \( \rho_2 \) be two one-dimensional Kohn-Sham densities. We define a linear interpolation between them as \( \rho_t = (1 - t)\rho_1 + t\rho_2 \) for all \( t \in [0,1] \). Although \( \sqrt{\rho_t} \in H^1(\mathbb{R}) \), it may not necessarily be \( v \)-representable or \( N \)-representable. We analyze three cases: 
(i) $\int \rho_1 = \int \rho_2 = 2N$ for \( N = 2, 3, 4 \); 
(ii) \( \int \rho_1 = 4 \) and \( \int \rho_2 = 6 \); and 
(iii) \( \int \rho_1 = 4 \) and \( \int \rho_2 = 8 \). Results for the first case are shown in Figure~\ref{fig:TN2}, \ref{fig:TN3}, and \ref{fig:TN4}. These results provide numerical evidence for the convexity of the functional \( T^+ \) with respect to $\rho$ on the convex set $\mathcal{I}_{2N}$ for $N=2,3,4$, respectively. This, however, does not establish convexity of $T^+$ by itself. Results for the second case are shown in Figure~\ref{fig:TN23}, illustrating concavity in direction from $\int \rho_t$ moving from $4$ to $8$. Finally, results for the third case are shown in Figure~\ref{fig:TN24}; a discontinuous jump around \( \int \rho_t = 6 \) occurs in this case. These results indicate that $T^+$ is not convex as a function of $\rho$, and is not in fact continuous, but when restricted to a fixed $N$, there is preliminary numerical evidence that it is convex over $\mathcal{I}_{2N}$. 
\begin{figure}[htbp]
    \centering
    \begin{subfigure}[b]{0.45\textwidth}
      \includegraphics[width=\textwidth]{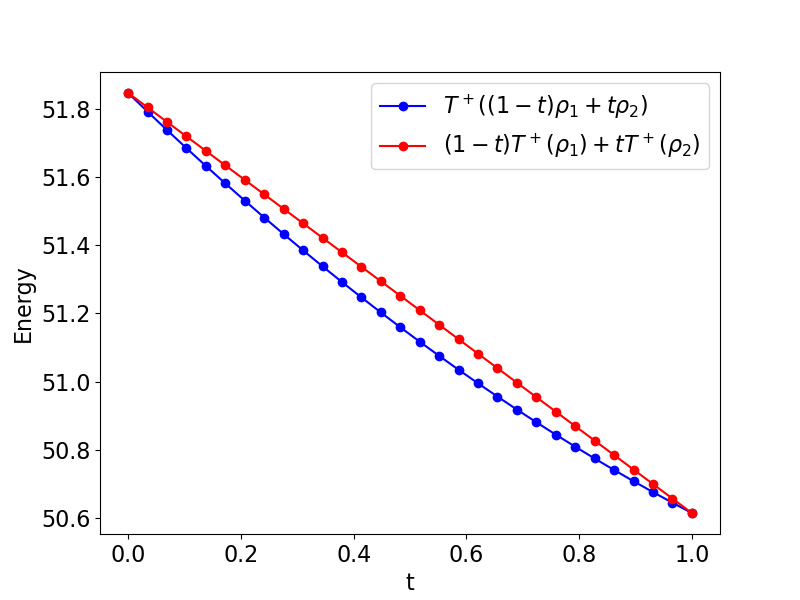}
    \end{subfigure}
    \begin{subfigure}[b]{0.45\textwidth}
        \includegraphics[width=\textwidth]{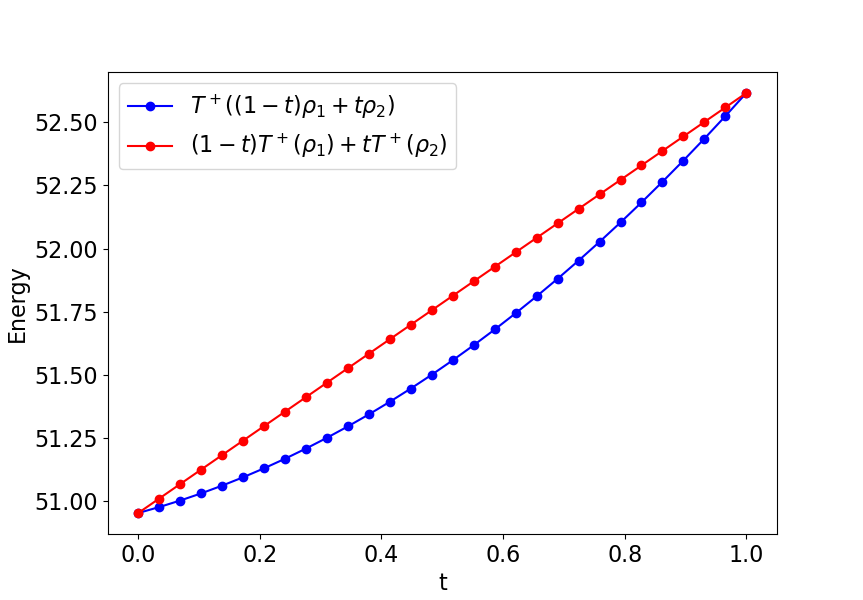}
    \end{subfigure}
    \begin{subfigure}[b]{0.45\textwidth}
        \includegraphics[width=\textwidth]{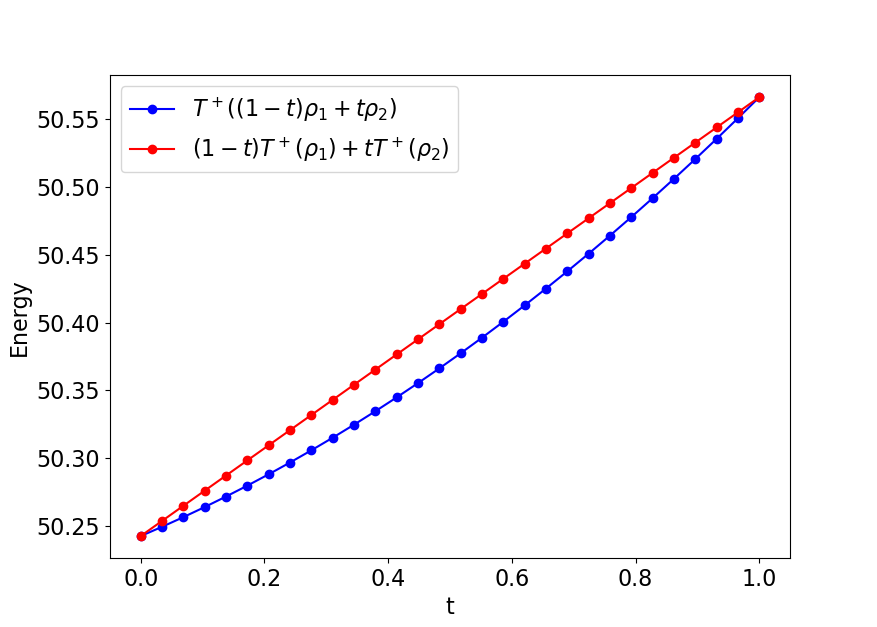}
    \end{subfigure}
    \begin{subfigure}[b]{0.45\textwidth}
        \includegraphics[width=\textwidth]{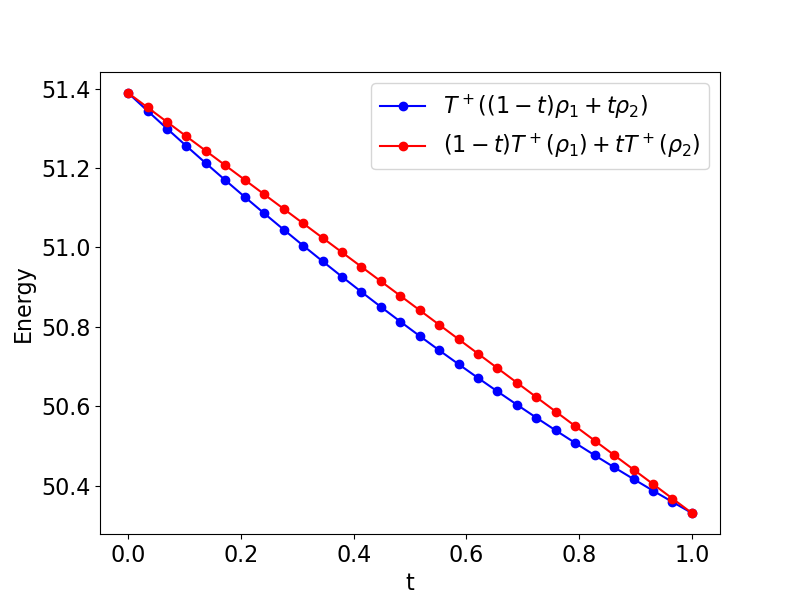}
    \end{subfigure}
    \caption{ Comparison of $T^+( (1-t)\rho_1 + t \rho_2)$ and $(1-t) T^+( \rho_1) + t T^+(\rho_2)$ when $\int_{\mathbb{R}}\rho_1 = \int_{\mathbb{R}}\rho_2 = 4  $ . }
    \label{fig:TN2}
\end{figure}

\begin{figure}[htbp]
    \centering
    \begin{subfigure}[b]{0.45\textwidth}
        \includegraphics[width=\textwidth]{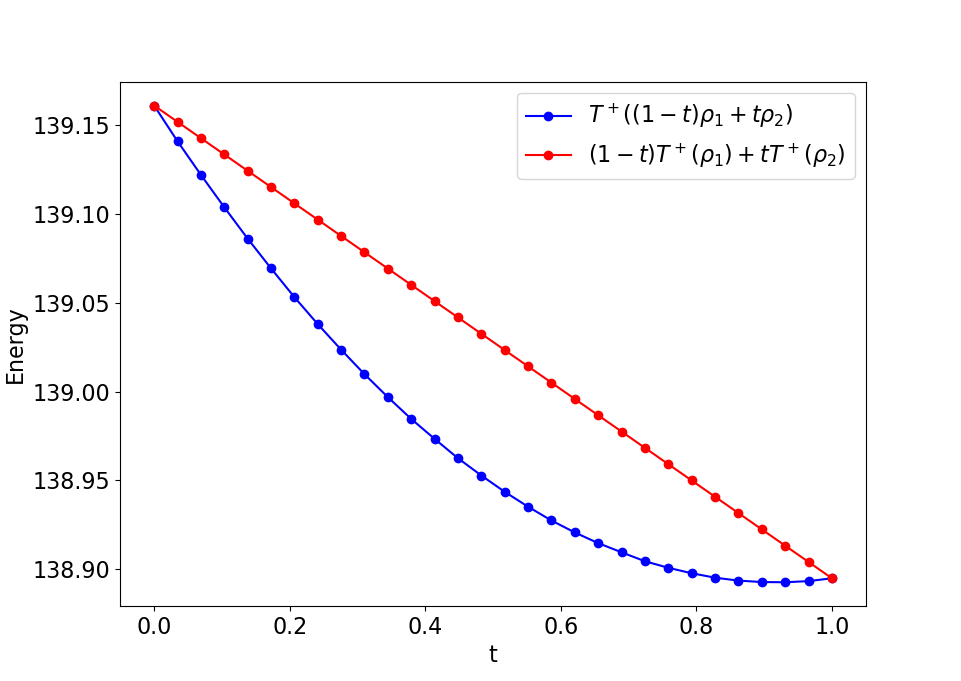}
    \end{subfigure}
    \begin{subfigure}[b]{0.45\textwidth}
        \includegraphics[width=\textwidth]{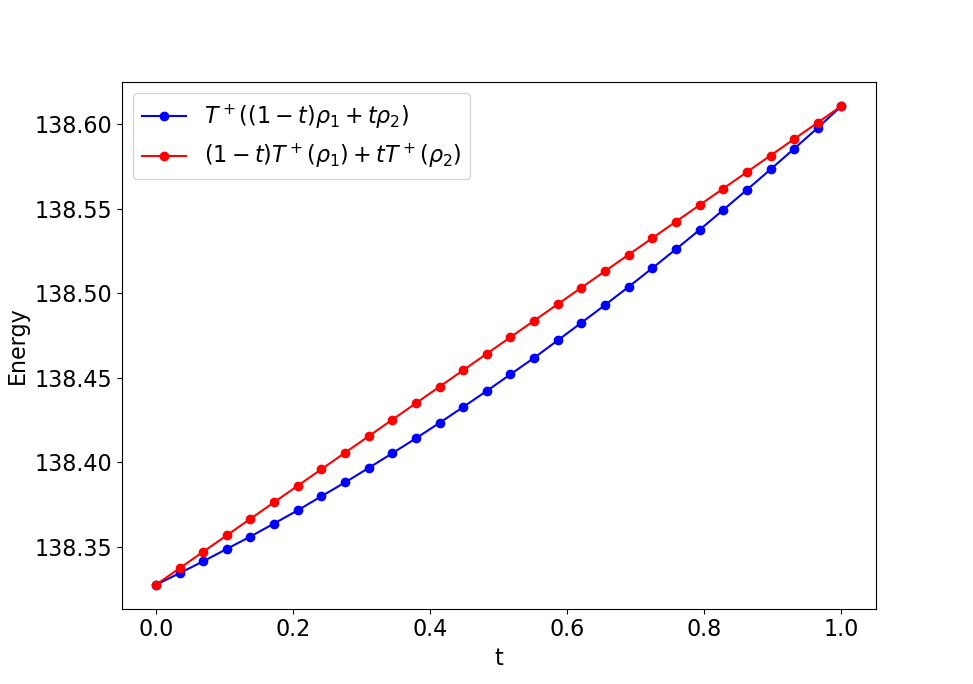}
    \end{subfigure}
    \begin{subfigure}[b]{0.45\textwidth}
        \includegraphics[width=\textwidth]{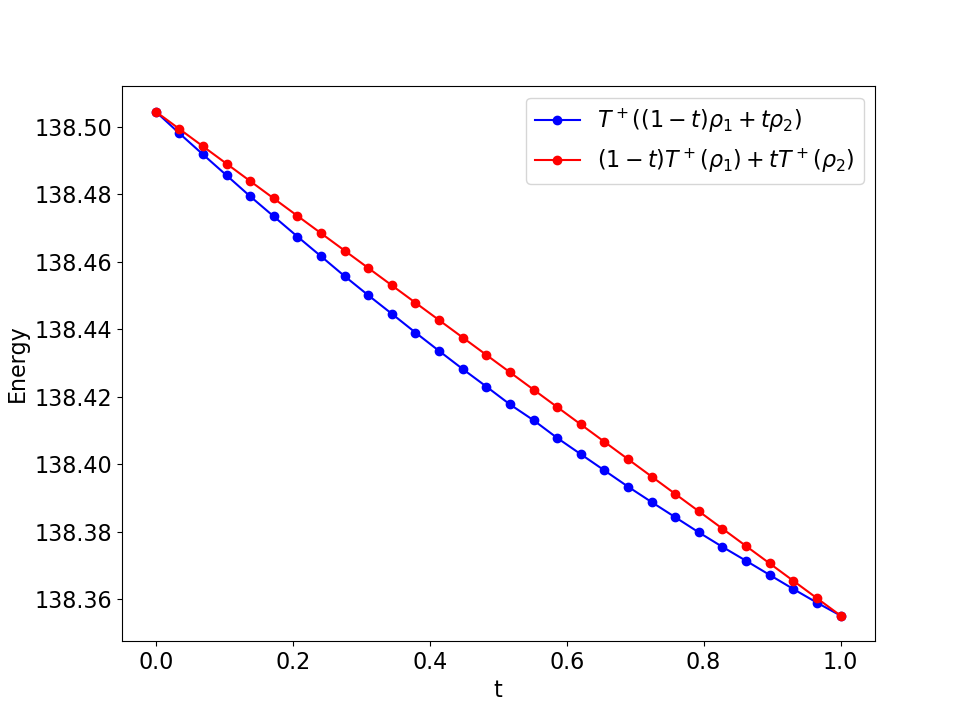}
    \end{subfigure}
    \begin{subfigure}[b]{0.45\textwidth}
        \includegraphics[width=\textwidth]{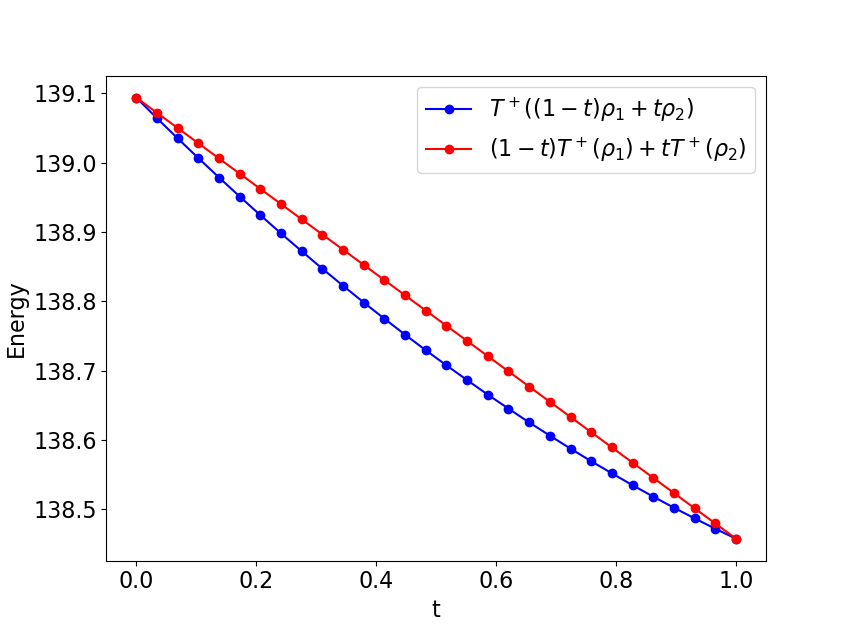}
    \end{subfigure}
    \caption{ Comparison of $T^+( (1-t)\rho_1 + t \rho_2)$ and $(1-t) T^+( \rho_1) + t T^+(\rho_2)$ when $\int_{\mathbb{R}}\rho_1 = \int_{\mathbb{R}}\rho_2 = 6$. }
    \label{fig:TN3}
\end{figure}

\begin{figure}[htbp]
    \centering
    \begin{subfigure}[b]{0.45\textwidth}
        \includegraphics[width=\textwidth]{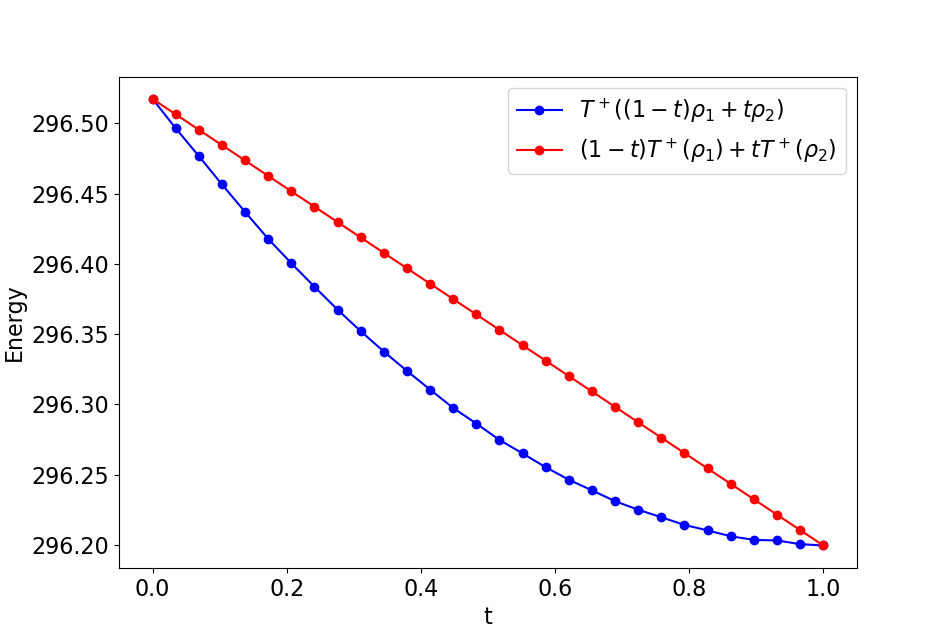}
    \end{subfigure}
    \begin{subfigure}[b]{0.45\textwidth}
        \includegraphics[width=\textwidth]{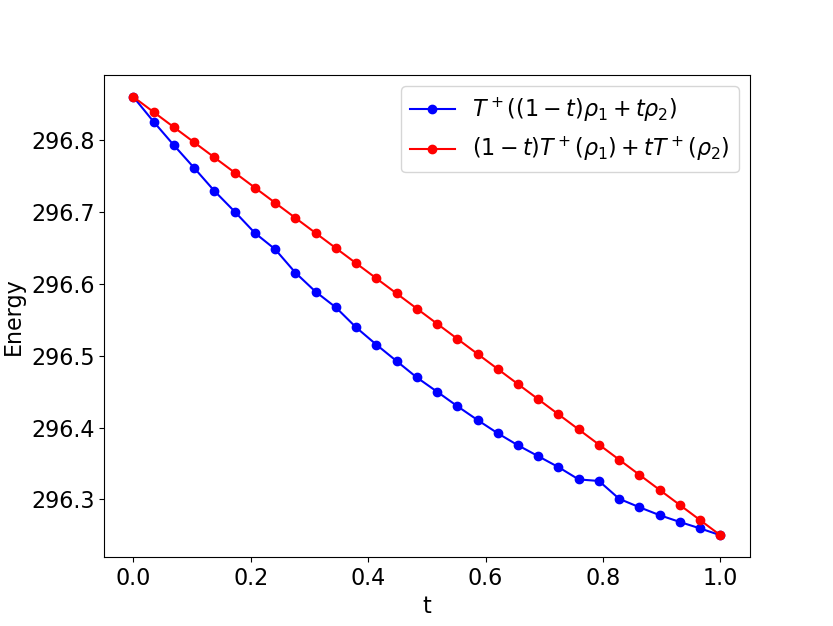}
    \end{subfigure}
    \begin{subfigure}[b]{0.45\textwidth}
        \includegraphics[width=\textwidth]{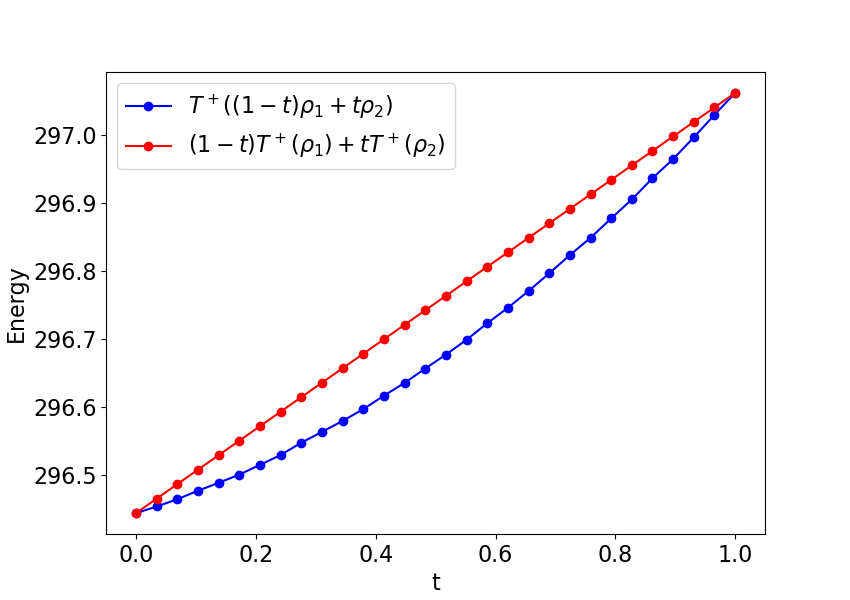}
    \end{subfigure}
    \begin{subfigure}[b]{0.45\textwidth}
        \includegraphics[width=\textwidth]{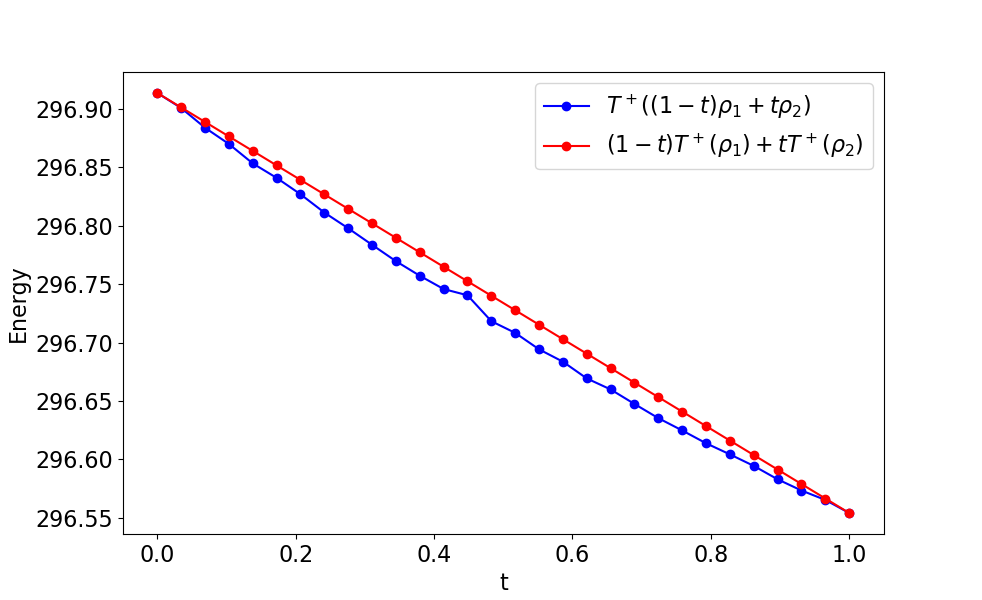}
    \end{subfigure}
    \caption{ Comparison of $T^+( (1-t)\rho_1 + t \rho_2)$ and $(1-t) T^+( \rho_1) + t T^+(\rho_2)$ when $\int_{\mathbb{R}}\rho_1 = \int_{\mathbb{R}}\rho_2 = 8$. }
    \label{fig:TN4}
\end{figure}

\begin{figure}[htbp]
    \centering
    \begin{subfigure}[b]{0.45\textwidth}
        \includegraphics[width=\textwidth]{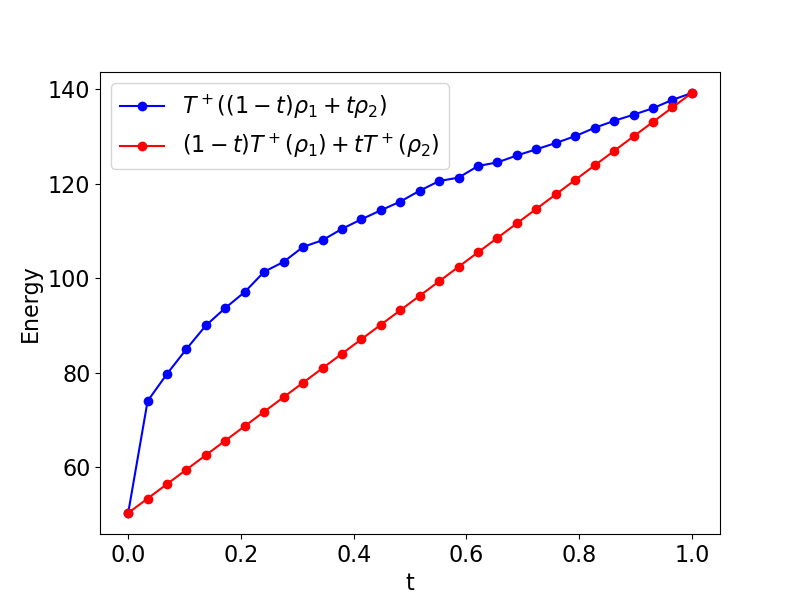}
    \end{subfigure}
    \begin{subfigure}[b]{0.45\textwidth}
        \includegraphics[width=\textwidth]{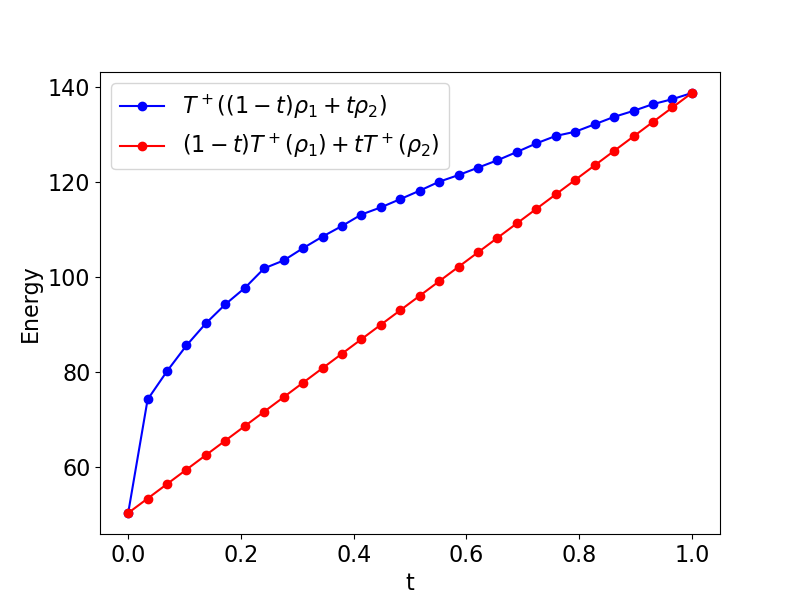}
    \end{subfigure}
    \begin{subfigure}[b]{0.45\textwidth}
        \includegraphics[width=\textwidth]{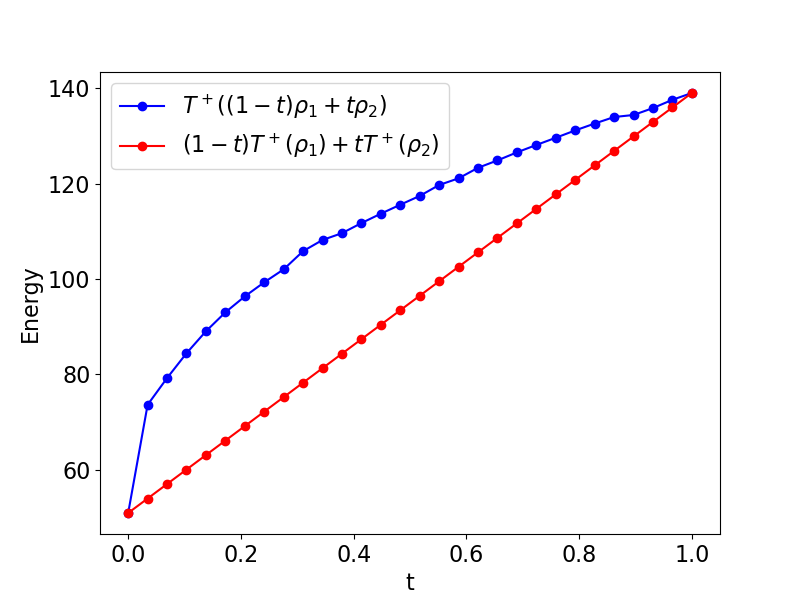}
    \end{subfigure}
    \begin{subfigure}[b]{0.45\textwidth}
        \includegraphics[width=\textwidth]{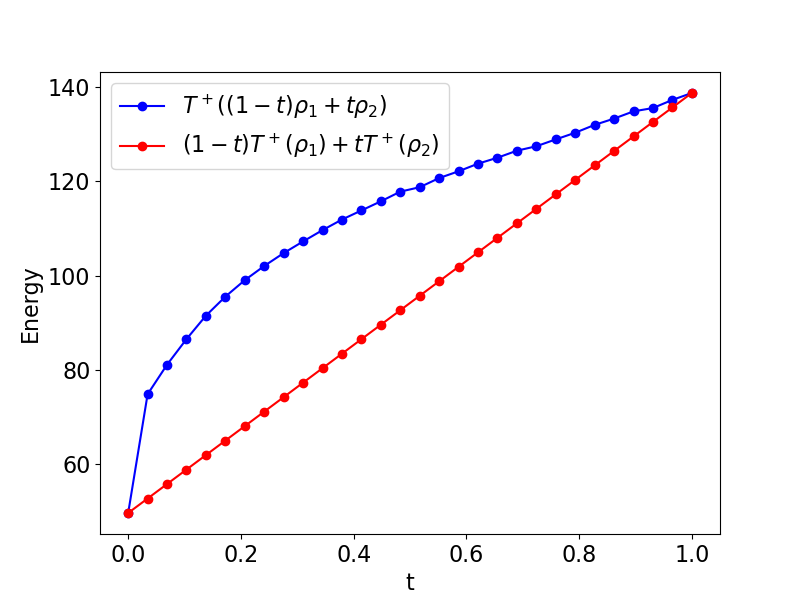}
    \end{subfigure}
    \caption{ Comparison of $T^+( (1-t)\rho_1 + t \rho_2)$ and $(1-t) T^+( \rho_1) + t T^+(\rho_2)$ when $\int_{\mathbb{R}}\rho_1 = 4$ and $\int_{\mathbb{R}}\rho_2 = 6$. }
    \label{fig:TN23}
\end{figure}

\begin{figure}[htbp]
    \centering
    \begin{subfigure}[b]{0.45\textwidth}
        \includegraphics[width=\textwidth]{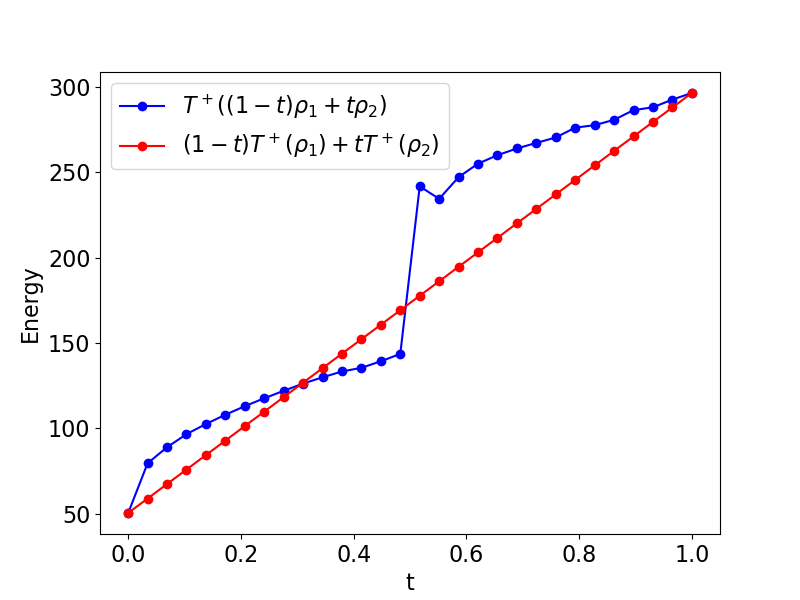}
    \end{subfigure}
    \begin{subfigure}[b]{0.45\textwidth}
        \includegraphics[width=\textwidth]{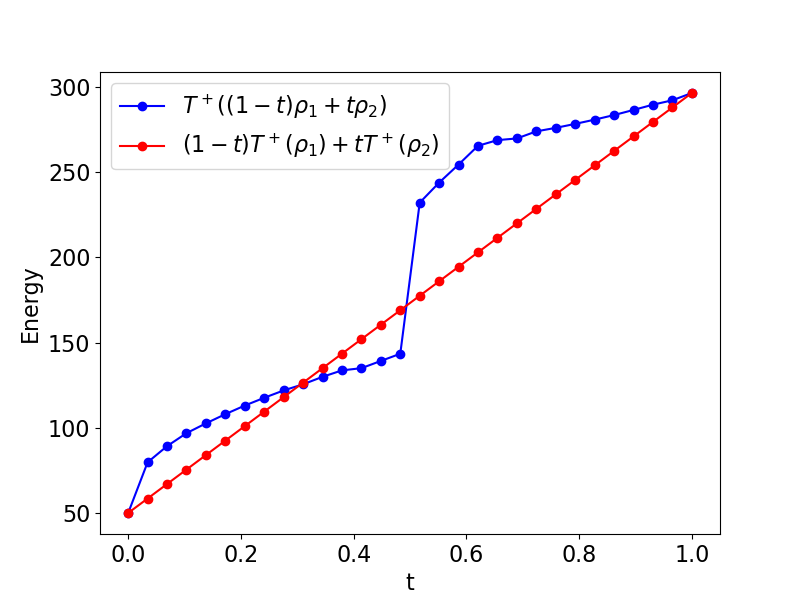}
    \end{subfigure}
    \begin{subfigure}[b]{0.45\textwidth}
        \includegraphics[width=\textwidth]{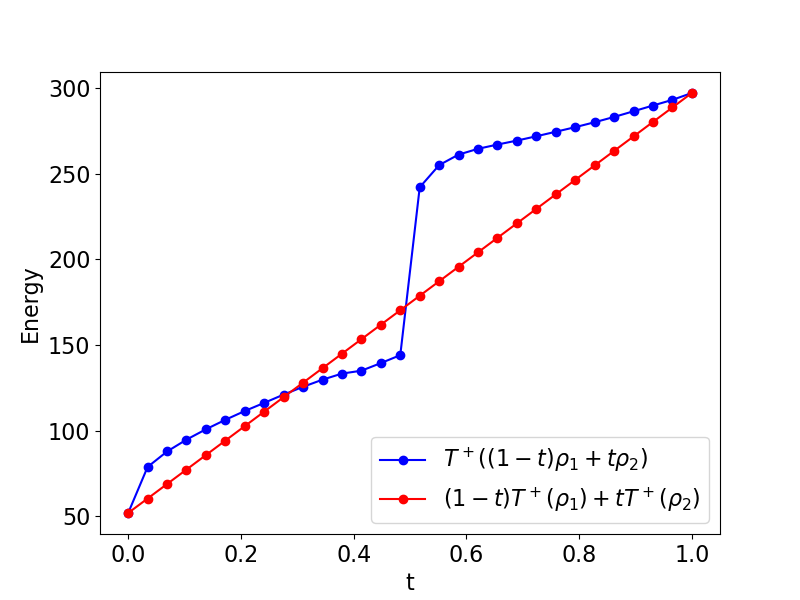}
    \end{subfigure}
    \begin{subfigure}[b]{0.45\textwidth}
        \includegraphics[width=\textwidth]{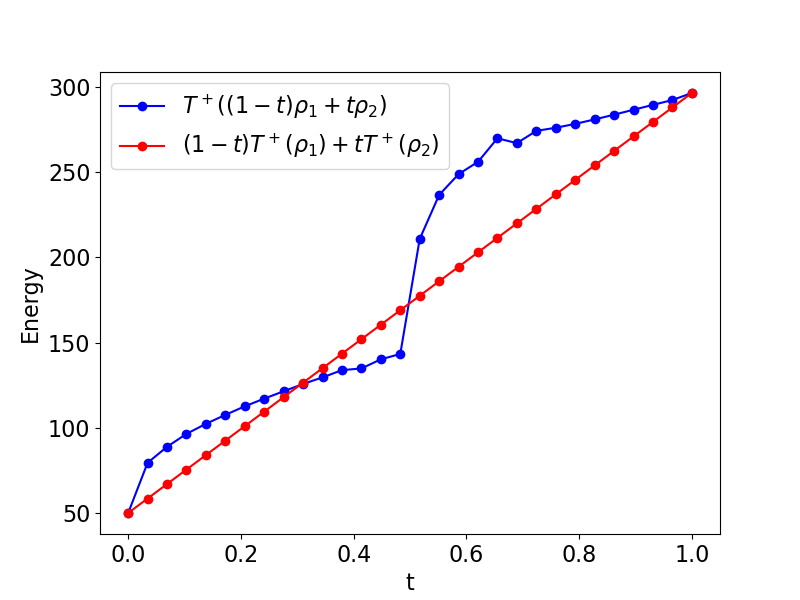}
    \end{subfigure}
    \caption{ Comparison of $T^+( (1-t)\rho_1 + t \rho_2)$ and $(1-t) T^+( \rho_1) + t T^+(\rho_2)$ when $\int_{\mathbb{R}}\rho_1 = 4$ and $\int_{\mathbb{R}}\rho_2 = 8$. }
    \label{fig:TN24}
\end{figure}

\section{Discussion and Conclusion}  
We analytically extended the restricted Kohn-Sham kinetic energy functional \( T_{RKS}(\rho) \) to $T(u)$. We then established that the proposed functional \( T^+(u) \) is an extension of $T(u)$. Specifically, we showed that \( T^+(u) = T_{RKS}(u^2) \) when \( u \) corresponds to the ground-state square root electron density of the Kohn-Sham system. 

The analytical construction in Section~\ref{sec:analytical_example} provides
an explicit illustration of the density-to-potential map enabled by the
angle-field formulation. For the family $\rho_{A,c}$, the prescribed density
determines the cumulative-density angle field, from which an orthonormal pair
of orbitals is obtained explicitly. These orbitals are then shown to be the
two lowest eigenfunctions of a common local potential, establishing that the
stationary angle field is the global minimizer. Consequently, the
corresponding four-electron ground-state densities, modified P\"oschl--Teller
potentials, orbital eigenvalues, and kinetic energies are all obtained in
closed form. In particular,
$T^+(u_{A,c})=T_{\mathrm{RKS}}(\rho_{A,c})=5A\pi/8$.

The direct evaluation of \(T_{RKS}(\rho)\) is delicate because the constraint
\(\rho = 2\sum_i|\phi_i|^2\) must be enforced in the minimization. One natural way to enforce this constraint is through Lagrange multipliers. According to Lagrange multipliers on Banach spaces theorem (see
Proposition~\ref{prop: lag_multiplier_on_banach}), existence of a
multiplier requires that the Fr\'echet derivative of the constraint at the minimizer must be surjective. It is generally not possible to show surjectivity of the derivative for the constraint
\(\rho = 2\sum_i|\phi_i|^2\) for an arbitrary $\rho$. When a multiplier does exist, it appears in the optimality condition through the adjoint of the derivative of the constraint, 
which leads to the corresponding adjoint problem. Moreover, even when a multiplier exists, it maybe an unbounded function which could have its own numerical difficulties. Our parametrization avoids these issues by enforcing the
constraint by construction, thereby eliminating the need to introduce and
solve the adjoint problem explicitly. Our approach overcomes this issue by using a parametric representation that avoids the need to explicitly solve the adjoint problem. The primary contribution of this work is Theorem~\ref{thm:energy_derivation}, where we show that the original minimization problem for the non-interacting kinetic energy is equivalent to a minimization problem involving angle fields that define the parametric representation referenced above. 

The central limitation of the new approach presented here is the restriction to closed shell systems. The key advantage here is that the orbitals can be chosen to be real-valued, and this significantly simplifies the mathematical analysis. The extension to open shell systems requires a careful consideration of spin orbitals, and will be explored in a future work. An interesting observation in this regard is that the functional $T^+(\rho)$ introduced in this work provides a numerical value for the kinetic energy even $\int \rho$ is an odd integer. Though this value is not physically meaningful, it can be construed as an upper limit to the non-interacting kinetic energy of the corresponding open shell system. 

We solve the minimization problem for the non-interacting kinetic energy by approximating the angle fields using radial basis functions, and using a standard minimization algorithm---Sequential Least Squares Quadratic Programming (SLSQP)---with randomly generated initial conditions to solve the resulting finite dimensional optimization problem. We observed, in particular, that the initial conditions corresponding to the orbitals of a particle in a box consistently produced results with the lowest energy obtained from a set of randomly chosen initial conditions. The kinetic energy per electron computed using our algorithm were within $1$ kcal/mol of the corresponding DFT calculations, demonstrating that our variational framework provides a reliable and accurate approach for computing the non-interacting kinetic energy functional.

We further examined the convexity properties of the proposed non-interacting kinetic energy functional---an analysis that, to the best of our knowledge, is not feasible with existing numerical methods. Within the convex set $\mathcal{I}_N$ for fixed particle numbers $N=2,3,4$, the functional appears to exhibit convex behavior. In contrast, concavity and/or discontinuities are observed when tested over densities with varying values of $N$. 

The numerical demonstrations presented here are limited to one dimension, but the algorithm is readily generalizable to three dimensions---this will be pursued in a future work. The particular algorithm that we used in this work does not scale well with increasing \( N \)---this is another avenue for future improvements. Another major contribution of this work, however, is a new and reliable method to generate datasets consisting of pairs of densities and corresponding non-interacting kinetic energies without resorting to a Kohn-Sham system. Such a dataset is likely to provide a better starting point to create transferable machine learning models for the non-interacting kinetic energy functional since it provides greater control over the distribution of input densities \( \rho \) and allows for the generation of more diverse training data. Learning such a non-interacting kinetic energy, and implementing it within the context of an orbital-free density functional theory framework, will be pursued in the future.

\section*{Acknowledgments}
The authors gratefully acknowledge Vikram Gavini, Bikash Kanungo, and Sambit Das for stimulating discussions on DFT. They also thank the anonymous referees for their constructive comments and suggestions.

\appendix

\section{Discontinuity of angle-fields}
\label{app:discontinuity}
\subsection{Discontinuity during inversion}
Let $w:\mathbb{R}\to\mathbb{R}$ be defined as
\begin{equation}
    w(x) = |x|\eta(x),
\end{equation}
where $\eta\in C_c^\infty(\mathbb{R})$ is a smooth cutoff with $\eta\equiv 1$ on $[-1/2,1/2]$ and $\operatorname{supp}(\eta)\subset[-1,1]$. Then $w\in H^1(\mathbb{R})$, $w(0)=0$, and $w\geq 0$. Define
\begin{equation}
    \phi_1(x) = c\,w(x)\,\mathbf{1}_{[0,\infty)}(x), \quad 
    \phi_2(x) = c\,w(x)\,\mathbf{1}_{(-\infty,0]}(x),
\end{equation}
with the normalization constant
\begin{equation}
    c = \frac{1}{\sqrt{\int_0^\infty w(x)^2\,dx}}.
\end{equation}
By construction, $\phi_1,\phi_2\in H^1(\mathbb{R})$ and $\phi_1, \phi_2\geq 0$. Moreover,
\begin{equation}
    \int_{\mathbb{R}} \phi_1(x)^2\,dx = 1, \qquad 
    \int_{\mathbb{R}} \phi_2(x)^2\,dx = 1, \qquad 
    \int_{\mathbb{R}} \phi_1(x)\phi_2(x)\,dx = 0.
\end{equation}
Thus $\{\phi_1,\phi_2\}$ forms an orthonormal set in $L^2(\mathbb{R})$. Setting
\begin{equation}
    u(x) = \sqrt{2\phi_1(x)^2+2\phi_2(x)^2} = c\,w(x),
\end{equation}
we may formally write
\begin{equation}
    \phi_1(x) = \frac{u(x)}{\sqrt{2}}\sin(\theta(x)), \quad \phi_2(x) = \frac{u(x)}{\sqrt{2}}\cos(\theta(x)).
\end{equation}
Note that $u \in H^1(\mathbb{R})$. 
Pointwise, the angle $\theta$ is determined as
\begin{equation}
    \theta(x) = 
    \begin{cases}
        0, & x < 0, \\
        \pi/2, & x > 0,
    \end{cases}
\end{equation}
with any arbitrary value at $\theta(0)$. Hence $\theta$ has jump discontinuity at $x=0$ and is not in $H^1_{loc}(\mathbb{R})$. 

\subsection{Discontinuity due to incompatibility}
Let $\rho: [0,1] \to \mathbb{R}$ be defined as
\begin{equation}
    \rho(x) = 4(\sin^2(2\pi x) + \sin^2(3\pi x)).
\end{equation}
Define the functions
\begin{equation}
    \phi_1 = \sqrt{2} \sin(2\pi x), \quad \phi_2 = \sqrt{2} \sin(3\pi x) ,
\end{equation}
We observe that $\sqrt{\rho} \in H^1([0,1])$, and similarly, $\phi_1, \phi_2 \in H^1([0,1])$. Furthermore, $\phi_1$ and $\phi_2$ form an orthonormal set and $\phi_1^2 + \phi_2^2 = \rho/2$. Notice that
\begin{equation}
    \phi_1 = \sqrt{\frac{\rho}{2}} \sin(\theta), \quad \phi_2 = \sqrt{\frac{\rho}{2}} \cos(\theta),
\end{equation}
and hence that $\theta$ is uniquely determined using the $\text{atan2}$ function as
\begin{equation}
    \theta = \text{atan2}(\phi_1, \phi_2).
\end{equation}
A plot of $\theta(x)$ is shown in Figure~\ref{fig:theta_plot}. Observe that $\theta$ is uniquely defined but exhibits jump discontinuity. Consequently, $\theta$ is not weakly differentiable.
\begin{figure}[h]
    \centering  \includegraphics[width=0.7\textwidth]{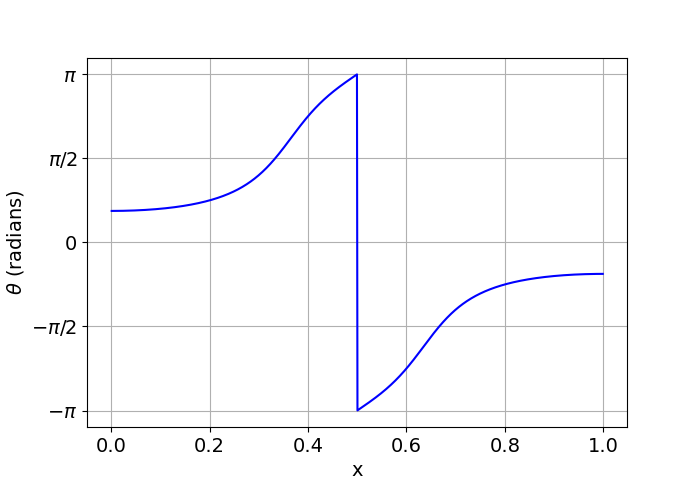}  
    \caption{Plot of $\theta(x)$ over $[0,1]$ showing discontinuity.}
  \label{fig:theta_plot}
\end{figure}

\section{Supplementary Propositions} \label{app:math_props}
\vspace{-0.2cm}
\begin{proposition}
\label{prop: lag_multiplier_on_banach} 
\textit{\textup{[}Section $4.14$ in \cite{Zeilder95}\textup{]}}
    Let \( X \) and \( Y \) be real Banach spaces and $U$ be an open subset of \( X \). Let \( f: U \to \mathbb{R} \) and \( g: U \to Y \) be continuously Frechet differentiable functions. Suppose \( u^* \in U \) is a local extremum of \( f \) subject to the constraint \( g (x)= 0 \). If the Frechet derivative \( Dg(u^*): X \to Y \) is a surjective linear map, then there exists a continuous linear functional \( \lambda^* \in Y^* \), dual to the space \( Y \), such that
$$Df(u^*) = \lambda \circ Dg(u^*), $$
where \( Df(u^*): X \to \mathbb{R} \) is Frechet derivative of $f$ at $u^*$. 
\end{proposition}

\vspace{-0.3cm}
\begin{proposition}
    \label{prop: Lieb_Schrodinger_GS}\textit{\textup{[}Theorem $11.8$ in \cite{Lieb01}\textup{]}}
    Let $V \in L^1_{loc}(\mathbb{R}^d)$, $\psi_0 \in H^1(\mathbb{R}^d)$ be a minimizer of $\mathcal{E}(f) = \int_{\mathbb{R}^d} \| \nabla f\|^2 + \int_{\mathbb{R}^d}Vf^2$ over $H^1(\mathbb{R}^d)$ with constraint $\|f\|_{L^2(\mathbb{R}^d)} =1$, and $V|\psi_0|^2$ be summable. If $V_+ \in L^\infty(\mathbb{R}^d)$, then $\psi_0$ is unique up to sign change and can be chosen to be strictly positive.
\end{proposition}

\vspace{-0.3cm}
\begin{proposition}
\label{prop:H1_level_set}\textit{\textup{[}Theorem $6.9$ in \cite{Lieb01}\textup{]}}
Let $f \in H^1(\mathbb{R}^d)$ and let $B = f^{-1}(\{0\})$. Then $\nabla f = 0$ for almost every $x \in B$. 
\end{proposition}

\section{Gradients for the RBF approximation} \label{app:FCgrad}
\vspace{-0.2cm}
The derivatives required for a gradient-based optimization algorithm are provided below for quick reference:
\[
\begin{split}
    \frac{\partial F}{\partial t_{ij}} &= \int_0^1 h \sigma_j - (\rho \cos^2{\theta_1}\cos{\theta_2}^2 \dots \cos{\theta_{i-1}^2} \theta_i^{'})\sigma_j^{'}, \\
    h &= \rho \cos^2{\theta_1} \dots \cos{\theta_{i-1}^2} \cos{\theta_i} \sin{\theta_i}\\
    &\Big((\theta_{i+1}^{'})^2 + \cos^2{\theta_{i+1}} (\theta_{i+2}^{'})^2 + \ldots + \cos^2{\theta_{i+1}} \dots \cos^2{\theta_{N-2}} (\theta_{N-1}^{'})^2 \Big), \\
    \frac{\partial C_{kl}}{\partial t_{ij}} &= \frac{\partial D_{kl}}{\partial t_{ij}} + \frac{\partial D_{lk}}{\partial t_{ij}}, \\
    \frac{\partial D_{kl}}{\partial t_{ij}} &= \int_0^1 g_{ik}(\Theta)\sigma_j \varphi^\Theta_l.
\end{split}
\]
In the final equation shown above, it is important to note that:
\[
\frac{\partial \varphi^\Theta_k }{\partial t_{ij}} = g_{ik}(\Theta)\sigma_j,
\]
where the following conditions hold, for \(1 \leq k \leq N-1\):
\begin{itemize}
    \item When \(k < i\), \(g_{ik}(\Theta) = 0\), 
    \item When \(k = i\), we have
\[
g_{ik}(\Theta) = \sqrt{\frac{\rho}{2}} 
\prod_{m=1}^{k} \cos(\theta_m),
\]
\item When \(k > i\), we have
\[
g_{ik}(\Theta) = -\sqrt{\frac{\rho}{2}}\sin(\theta_i) \sin(\theta_k)\prod_{m=1}^{i-1}\cos(\theta_m)\prod_{p=i+1}^{k-1}\cos(\theta_p).
\] 
\end{itemize}
Finally, when \(k = N\), we have
\[
g_{iN}(\Theta) = -\sqrt{\frac{\rho}{2}} \sin(\theta_i) \prod_{m=1}^{i-1}\cos(\theta_{m}) \prod_{p=i+1}^{N-1}\cos(\theta_{p}).
\]

\bibliographystyle{amsalpha}
\bibliography{tks_ref}

\end{document}